\documentclass[10pt, a4paper]{article}

\usepackage{lrec-coling2024} 
\usepackage{amsmath}
\usepackage{times}  
\usepackage{helvet}  
\usepackage{courier}  

\usepackage{natbib}  
\usepackage{caption} 
\usepackage{algorithm}
\usepackage{amsmath}
\usepackage{amssymb}
\usepackage{graphicx}
\usepackage{amsfonts}
\usepackage{algpseudocode}
\usepackage{ragged2e}
\usepackage{bbm,dsfont}
\usepackage{multirow}
\usepackage{footnote}
\makesavenoteenv{tabular}

\makeatletter
\newcommand*\bigcdot{\mathpalette\bigcdot@{.5}}
\newcommand*\bigcdot@[2]{\mathbin{\vcenter{\hbox{\scalebox{#2}{$\m@th#1\bullet$}}}}}
\makeatother
\usepackage{newfloat}
\usepackage{listings}
\DeclareCaptionStyle{ruled}{labelfont=normalfont,labelsep=colon,strut=off} 
\floatstyle{ruled}
\newfloat{listing}{tb}{lst}{}
\floatname{listing}{Listing}
\newtheorem{theorem}{Theorem}[section]

\newtheorem{proposition}[theorem]{Proposition}

\newtheorem{defi}[theorem]{Definition}

\newtheorem{rema}[theorem]{Remark}
\newtheorem{exam}[theorem]{Example}

\newenvironment{definition}{\begin{defi}\rm}{\hfill $\lhd$\end{defi}}

\newenvironment{proof}{\begin{trivlist}\item[]{\bf
Proof.}}{\hfill {\sc qed}\end{trivlist}}

\newtheorem{claim2}{\sc Claim}

\title{Locally Differentially Private In-Context Learning}

\name{\begin{tabular}{c}
Chunyan Zheng, Keke Sun, Wenhao Zhao, Haibo Zhou,\\ 
Lixin Jiang, Shaoyang Song, Chunlai Zhou 
\end{tabular}} 

\address{Renmin University of China,
         Beijing, CHINA \\
         jany.zh666@gmail.com, 
         \{skk2020, zhaowh, zhoub21,lixinjiang,songshaoyang,czhou\}@ruc.edu.cn\\}

\abstract{
Large pretrained language models (LLMs) have shown
surprising In-Context Learning (ICL) ability. An important application in deploying large language models is  to augment LLMs with a private database for some specific task.  The main problem with this promising   commercial use is that LLMs have been shown to memorize their training data and their prompt data are vulnerable to membership inference attacks (MIA) and prompt leaking attacks. 
In order to deal with this problem, we treat LLMs as \emph{untrusted in privacy} and  propose a \emph{locally differentially private framework of in-context learning} (LDP-ICL) in the settings where labels are sensitive.  Considering the mechanisms of in-context learning
in Transformers by gradient descent, we provide an analysis of the trade-off between privacy and utility in such LDP-ICL for classification. Moreover, we apply LDP-ICL to the 
discrete distribution estimation problem. In the end, we perform several experiments to demonstrate our analysis results. 
 \\ \newline  }

\begin{document}

\maketitleabstract

\section{Introduction} \label{sec:introduction}
\begin{figure*}[t]
    \centering
    \includegraphics[width=0.84\linewidth,height=6.2cm]{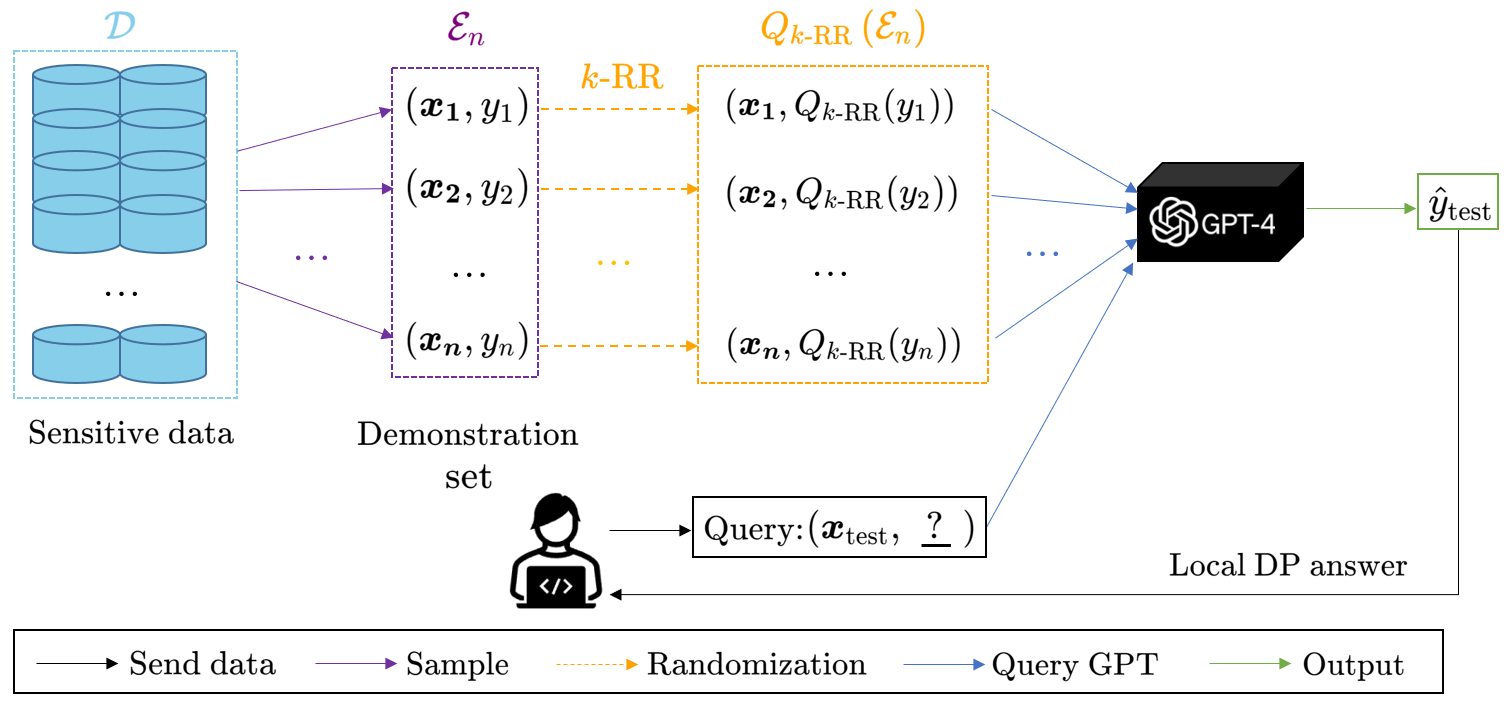}
    \captionsetup{justification=raggedright,singlelinecheck=false}
    \caption{The framework of LDP-ICL: We first sample a few input-label pairs from the original private database to form the demonstration set. Next we employ the $k$-ary randomized response mechanism $Q_{k\text{-RR}}$ to perturb the labels and then perform the ICL with a given query $\boldsymbol{x_{\text{test}}}$ prepended by the noisy demonstration set. At the end, the response is returned to the adversary.}
    \label{fig:main}
\end{figure*}

Large language models (LLMs) have exhibited surprising emergent abilities for in-context learning \citep{Brown2020language}. With a few  input-label pairs as exemplars, they can predict the label for an unseen input without additional parameter modifications. Although the training data for an LLM is usually assumed to be public and non-private, the demonstration pairs in in-context learning for the downstream task may contain private information about individual users and  are often considered to be \emph{sensitive}.  After Samsung leaked private data by using the LLM ChatGPT \citep{Mitchell2023Samsung} and Italy banned the use of ChatGPT due to the concern about the exposure of personal information,  it becomes imminent to study the privacy-preservation for the LLMs.  

In this paper, we propose a \emph{locally differentially private in-context learning} (LDP-ICL).  Differential privacy (DP) is now  a gold standard of privacy-preserving which addresses the paradox of learning nothing about an individual while learning useful information about a population \citep{Dwork2006calibrating}. There are two kinds of DP models: one is the central model and the other is the local model \citep{Kasiviswanathan2011can,Duchi2013local,Warner1965randomized}.  In the central model, the original private data are aggregated by the curator and then are perturbed by a DP mechanism before publishing. On the other hand, in the local model, the private data of each individual get randomized locally according to a DP mechanism and then are aggregated by the curator. The main difference between these two models is that the local model treats the data curator \emph{untrusted} while the central model believes in the curator.  In the in-context learning, LLMs are usually the data aggregator and  LLMs have been shown to memorize their training data \citep{Biderman2023emergent,Carlini2019secret},  and their prompt data are vulnerable to membership inference attacks (MIA) \citep{Duan2023flocks} and prompt leaking attacks \cite{Perez2022ignore}.  In this sense we consider LLMs as untrusted. Our first contribution is to propose a locally differentially private  mechanism for protecting individuals' privacy in ICL.  In this paper, we focus on the classification problem, especially the binary classification.  For each input-label pair in the demonstration set, we consider the input as an identifier and hence nonsensitive but regard the label as \emph{sensitive} \citep{Dinur2003revealing}. 
We employ the well-known LDP mechanism \emph{k-ary randomized response (k-RR)} to perturb each label and obtain an input-label pair with a noisy label \citep{Kairouz2016discrete,Wang2017locally}.  An adversary can query the LLM  with $\boldsymbol{x_{\text{test}}}$ prepended by such a perturbed demonstration set. Due to the noises in the labels in the demonstration set, the adversary obtains a corresponding noisy label as response  to the input $\boldsymbol{x_{\text{test}}}$, from which he cannot reliably tell the true label of any input in the private demonstration set. This implies that the privacy in the labels are protected. The process of such an LDP-ICL is illustrated in Figure \ref{fig:main}. 

Our second contribution is to propose a formula to represent the prediction output probability for the noisy label of the query in the LDP-ICL (Eq. (\ref{eq:LDP-ICL}))  by considering ICL for classification as an \emph{implicit} gradient-descent based optimization, which is the dual form of the Transformer attention in ICL \citep{Irie2022dual,von2023transformers,Dai2022can}. From this formula, we obtain the trade-off between privacy-preservation and utility, which is measured by the accuracy rate of the query answers.  When the privacy-preservation gets stronger, i.e., $\epsilon$ in $k$-RR gets smaller,  the accuracy  becomes smaller. Moreover, we run experiments on several datasets for the classification task and demonstrate the trade-off effects (Figure \ref{fig:classification}).  
In order  to support our understanding of LDP-ICL,  we apply it to a \emph{touch-stone} LDP problem: the discrete distribution estimation problem \citep{Kairouz2016discrete,Wang2017locally}. We design an algorithm with LDP-ICL to perform distribution estimation of sensitive labels in the original private database  (Algorithm \ref{algorithm:ldp-icl-estimation}) and compare the results with the classic Warner's mechanism for the same task \citep{Warner1965randomized}. Our results show that our algorithm performs better than Warner's mechanism in the privacy-utility trade-off for the high-privacy region. 

The rest of the paper is organized as follows. In Section \ref{sec:basics}, we present the definition of in-context learning. In Section \ref{sec:main}, we consider ICL as implicit  gradient-descent optimization which is a dual form of Transformers attention mechanism and propose LDP-ICL. And we analyse the trade-off between privacy and utility in the LDP-ICL and further deal with the discrete distribution estimation problem with the analysis. We perform experiments to support our analysis in Section \ref{sec:experiments} and conclude with related works and further problems in Section \ref{sec:conclusions}. Besides, we  will add  the effects with demonstration sets of different sizes in the extended version, which also include some experiment details and proofs.

\section{In-context Learning} \label{sec:basics}

In this paper, we focus on  in-context learning(ICL) for classification tasks using large language models (LLMs) \citep{Brown2020language}. In-context learning is a paradigm that allows large language
models to learn tasks given only a few examples in the form of demonstrations, which is an emergent ability for LLMs. 
Essentially, it gauges the probability of a prospective answer based on the provided demonstrations, leveraging a well-trained large language model.
For a classification task, given a query input text $\boldsymbol{x_{\text{test}}}$ and a candidate answer 
set $Y=\{y_1,y_2,\dots,y_M\}$, we need to predict a label $\hat{y}_{\text{test}} \in Y$ conditional on $n$ demonstration examples $\mathcal{E}_n=\{I,s(\boldsymbol{x_1},y_1),\dots,s(\boldsymbol{x_n},y_n)\}$ or 
$\mathcal{E}_n=\{s(\boldsymbol{x_1},y_1),\dots,s(\boldsymbol{x_n},y_n)\}$, where $I$ is an optional task 
instruction and $s(\boldsymbol{x_i},y_i) (1\leq i \leq n)$ is an example written in natural language texts according to the task. Formally, given a GPT3.5 model $\mathcal{M}$, we calculate the probability for each answer $y_j$ :
 $     P_{\mathcal{M}}(y_j\mid \mathcal{E}_n,\boldsymbol{x_{\text{test}}}). $                     
Then, the ultimate predicted label $\hat{y}_{\text{test}}$ is given by the candidate answer with the highest probability:
                 $  \hat{y}_{\text{test}}= \operatorname{LLM}(\mathcal{E}_n, \boldsymbol{x}_{\text{test}})  = y_{\arg\max_{j}P_{\mathcal{M}}(y_j|\mathcal{E}_n,\boldsymbol{x_{\text{test}}})}.
                   $                      
For example, we could predict the class
label in a binary sentiment classification by comparing the prediction probability of the two labels: 0 and 1. 
The following are some characteristics which make ICL an important form of learning method: 
     without optimizing any parameters, ICL directly performs predictions on the pretrained language models; 
    by altering the demonstration and templates, it is easier to incorporate human knowledge into LLMs \citep{Wei2022chain};
     ICL is a training-free learning framework and can be easily applied to large-scale real-world tasks.
   

In this paper, we work with ICL for \emph{large} language models, which can override semantic priors from pretraining when presented with in-context demonstrations that contradict priors \citep{wei2023larger}. It performs both \emph{task recognition} for identifying tasks and \emph{task learning} for learning new input-label mappings from demonstrations  \citep{Pan2023context}.  So its performance improves consistently with more demonstrations.


\section{Locally Differentially Private ICL} \label{sec:main}

This section is our main contribution. First we adapt some previous results in the literature and obtain a formula for the prediction probability of the query answer in ICL by considering ICL as a dual form of gradient-descent-based optimization. 
Next we formulate locally differentially private ICL and use the above prediction formula to analyze the trade-off between privacy-preservation and prediction accuracy.  After that, we design an algorithm with LDP-ICL for the discrete distribution estimation problem and compare with the standard Warner's mechanism \citep{Warner1965randomized}.

\subsection{In-context Learning by Gradient Descent}

 Mathematically,  let $\boldsymbol{W}_0, \Delta \boldsymbol{W} \in \mathbb{R}^{N\times M }$ be the initialized weight and update matrix for a given classification task, respectively, where update is performed across few demonstrations: $\mathcal{E}_n=\{\left(\boldsymbol{x_i},\boldsymbol{y_i}\right)\}_{i=1}^n$ comprising input representations $\boldsymbol{x_i} \in \mathbb{R}^{N}$ and corresponding labels $\boldsymbol{y_i} \in \mathcal{Y} = \{y_1, \cdots, y_M\} \subseteq \mathbb{R}^{M}$. 
$\boldsymbol{W}_0 \boldsymbol{x}_{\text{test}}$ is the answer of the zero-shot learning, i.e., ICL with \emph{no demonstrations} and hence serves as an important reference.  In this paper, we will fix this formalization. 
According to \citep{pmlr-v162-irie22a}, transformer attention has a dual
form of gradient descent.
Gradient descent use back-propagation algorithm to calculate $\Delta \boldsymbol{W}$, by summing the outer products of $\{\boldsymbol{x_i}\}_{i=1}^n$ with their corresponding error signals $\mathbf{e}_i \in \mathbb{R}^{N \times M}$
\begin{align}
\Delta \boldsymbol{W}=\sum\limits_{i=1}^{n}\mathbf{e}_i\otimes\boldsymbol{x}_i= \sum\limits_{i=1}^{n}\mathbf{e}_i\boldsymbol{x}_i^T \label{dw}   
\end{align}

Then, given a specific query $\boldsymbol{x_{\text{test}}}$, we obtain its prediction
\begin{align}
    \hat{\boldsymbol{y}}_{\text{test}} = \left(\boldsymbol{W}_0+\Delta \boldsymbol{W}\right)\boldsymbol{x_{\text{test}}} \label{pre}
\end{align}

Combining (\ref{dw}) and (\ref{pre}), we derive the dual form of gradient decent
\begin{align*}
\hat{\boldsymbol{y}}_{\text{test}}&=\boldsymbol{W}_0\boldsymbol{x_{\text{test}}}+\sum\limits_{i=1}^{n}\mathbf{e}_i\boldsymbol{x}_i^T\boldsymbol{x_{\text{test}}}\\ &=\boldsymbol{W}_0\boldsymbol{x_{\text{test}}}+\sum\limits_{i=1}^{n}\mathbf{e}_i\left(\boldsymbol{x}_i^T\boldsymbol{x_{\text{test}}}\right) \\ 
&=\boldsymbol{W}_0\boldsymbol{x_{\text{test}}}+\text{LinearAttn}\left(\boldsymbol{E},\boldsymbol{X},\boldsymbol{x_{\text{test}}}\right)
\end{align*}
where Linear Attention operation is performed over error signal matrix $\boldsymbol{E}$, demonstration set $\boldsymbol{X}$ and query $\boldsymbol{x_{\text{test}}}$ representing values, keys and query, respectively.

We now illustrate that our ICL for classification tasks can be realized through self-attention mechanism \emph{followed by an activation function such as softmax or sigmoid}, also interpretable as an implicit gradient descent step on the cross-entropy(CE) loss. For simplicity and illustration, we focus on binary classification, where the 2-dimensional one-hot label vector can be treated as a real number in $\{0,1\}$.
Given demonstrations: $\mathcal{E}_n=\{\left(\boldsymbol{x_i},y_i\right)\}_{i=1}^n, y_i \in \{0,1\}$, the binary CE loss measures the dissimilarity between the predicted probability with the true binary labels
\begin{small}
\[
L\left(\boldsymbol{W}\right) = \sum\limits_{i=1}^n\left[y_i\ln\left[\sigma\left(\boldsymbol{W}\boldsymbol{x_i}\right)\right] + \left(1-y_i\right)\ln\left[1-\sigma\left(\boldsymbol{W}\boldsymbol{x_i}\right)\right]\right]
\]
\end{small}
where $\sigma(z) \triangleq \frac{1}{1 + e^{-z}} $ denotes sigmoid function and $\boldsymbol{W} \in \mathbb{R}^{N}$ is weight matrix.
Applying a single gradient descent iteration to the loss function $L$ with learning rate $\eta$ yields the weight change
\begin{small}
\begin{align}
    \Delta \boldsymbol{W} = -\eta \nabla_{\boldsymbol{W}} L\left(\boldsymbol{W}\right) = -\eta\sum\limits_{i=1}^n \left(\sigma\left(\boldsymbol{W}\boldsymbol{x_i}\right) - y_i\right) \boldsymbol{x_i}^T \label{delta_W}
\end{align}
\end{small}
Let $\hat{p}_{\text{test}}$ be the prediction probability of the true label for the query $\boldsymbol{x}_{\text{text}}$ in the zero-shot learning. 
Consequently, this alteration in weights will result in an update in the prediction $\hat{p}_{\text{test}}$ for query $\boldsymbol{x_{\text{test}}}$
\begin{small}
\begin{align}
   \begin{pmatrix}
\boldsymbol{x_{\text{test}}} \\ 
\hat{p}_{\text{test}}
\end{pmatrix}
&=
\begin{pmatrix}
\boldsymbol{x_{\text{test}}} \\ \sigma\left(\boldsymbol{W}_0\boldsymbol{x_{\text{test}}}\right)
\end{pmatrix} \nonumber \\
\leftarrow
\begin{pmatrix}
\boldsymbol{x_{\text{test}}} \\ \sigma\left(\left(\boldsymbol{W}_0+\Delta \boldsymbol{W}\right)\boldsymbol{x_{\text{test}}}\right)
\end{pmatrix}
&=
\begin{pmatrix}
\boldsymbol{x_{\text{test}}} \\ \sigma\left(\boldsymbol{W}_0\boldsymbol{x_{\text{test}}}+\Delta \boldsymbol{W}\boldsymbol{x_{\text{test}}}\right) \label{prediction}
\end{pmatrix}
\end{align}
\end{small}

Combining (\ref{delta_W}) and (\ref{prediction}), we rewrite the updated prediction as 
\begin{small}
\begin{align}
     \hat{p}_{\text{test}}^{\text{(u)}} & \triangleq \sigma\left(\boldsymbol{W}_{0}\boldsymbol{x_{\text{test}}}+\Delta \boldsymbol{W} \boldsymbol{x_{\text{test}}}\right) \nonumber \\
    & = \sigma\left(\boldsymbol{W}_0x_{\text{test}}-\eta\sum_{i=1}^{n}\left(\sigma\left(\boldsymbol{W}_0\boldsymbol{x_i}\right)-y_{i}\right)\boldsymbol{x_i}^{T}\cdot\boldsymbol{x_{\text{test}}}\right) \label{eq:GD}
\end{align}
\end{small}
Note that the gradient-descent step is performed on the inner transformer attention mechanism. Following \citep{Zhmoginov2022hypertransformer}, self-attention mechanism can emulate gradient descent on a classification task:

\begin{proposition} Given previous token: $\mathcal{E}_n=\{\left(\boldsymbol{x_i},y_i\right)\}_{i=1}^n$, we can construct key, query and value matrices $\boldsymbol{W_k}$, $\boldsymbol{W_q}$, $\boldsymbol{W_v}$ as well as the projection matrix $\boldsymbol{P}$ such that a 1-head linear attention operation on the matrix $X:= [\mathcal{E}_n, \boldsymbol{x_{\operatorname{test}}}]$ followed by sigmoid(or softmax)
 yields the same results $\hat{p}_{\operatorname{test}}^{\operatorname{(u)}}$ as induced by gradient descent
 \begin{small}
$$
\hat{p}_{\operatorname{test}}^{\operatorname{(u)}}  = \sigma(\boldsymbol{W}_0 \boldsymbol{x}_{\text{test}} + \boldsymbol{P}\operatorname{LinearAttn}(\boldsymbol{W_v} X,\boldsymbol{W_k} X,\boldsymbol{W_q} X))
$$
\end{small}
\end{proposition}


\begin{proof}
    For simplicity, we don't specify the sizes of different matrices. The context will determine the sizes.  

We define matrix $P$ and  operator $\sigma^-$  by 

\[ P = \eta I, \text{ and }
\sigma^-(  
\left[
  \begin{array}{c}
    \boldsymbol{W}_0 x_i   \\
     y_i  \\
 \end{array}
\right] )
=
\begin{bmatrix}
          0 \\
          \sigma(\boldsymbol{W}_0 x_i)-y_i
         \end{bmatrix}
\]

where $I$ is the identity matrix. Note that $\sigma^-$ is just a sigmoid function followed by a subtraction. Define
\[
\boldsymbol{W}_K = \boldsymbol{W}_Q =
\left[
  \begin{array}{cc}
     I & 0   \\
    0 & 0   \\
 \end{array}
\right],
\boldsymbol{W}_V= 
\begin{bmatrix}
         \boldsymbol{W}_0 & 0 \\
          0 & I 
         \end{bmatrix}
\]

Consider
\begin{align*}
P = \sum_{i=1}^n \Bigg[ &\sigma^-\left( \left[
  \begin{array}{cc}
     \boldsymbol{W}_0 &  0\\
   0  &  I\\
     \end{array}
\right]
\begin{bmatrix}
           x_i \\
           y_i 
         \end{bmatrix} \right) \otimes \\
&\left( \begin{bmatrix}
           I & 0 \\
           0 & 0 
         \end{bmatrix}
         \begin{bmatrix}
           x_i \\
           y_i 
         \end{bmatrix} \right) \Bigg] \cdot
    \left( \begin{bmatrix}
           I & 0 \\
           0 & 0 
         \end{bmatrix}
         \begin{bmatrix}
           x_{\text{test}} \\
           p_{\text{test}} 
         \end{bmatrix} \right)
\end{align*}

We can compute the above expression and obtain that it is equal to $\eta \sum_{i=1}^n (\sigma(\boldsymbol{W}_0 x_i)- y_i)) \boldsymbol{x_i}^T \boldsymbol{x}_{\text{test}}$, which is just $-\Delta \boldsymbol{W} x_{\text{test}}$.  Then the matrices $\boldsymbol{W}_{\boldsymbol{v}}, \boldsymbol{W}_{\boldsymbol{k}}$ and $\boldsymbol{W}_{\boldsymbol{q}}$ can be constructed from $\sigma^-$ and the above matrices $\boldsymbol{W}_V$, $\boldsymbol{W}_K$ and $\boldsymbol{W}_Q$  respectively. 

\end{proof}

This proposition can explain well why the performance of ICL for classification improves with more demonstrations with \emph{true} labels. 
We see from the formula Eq. (\ref{eq:GD}) that any demonstration with true label will increase the prediction probability of the query $\boldsymbol{x_{\text{test}}}$  and any exemplary with false label will decrease the prediction probability.  For the binary classification, when $\boldsymbol{x}_i^T \boldsymbol{x}_{\text{test}} > 0$ and $y_i=1$, then  $\left(\sigma\left(\boldsymbol{W}_0\boldsymbol{x_i}\right)-y_{i}\right)\boldsymbol{x_i}^{T}\cdot\boldsymbol{x_{\text{test}}}) < 0$.  Since $\sigma$ is an increasing function, the demonstration $(\boldsymbol{x}_i, y_i)$ contribute to increase the prediction probability. Other cases can be analyzed similarly. 
This may explain well the emergent ability of \emph{task learning} of the LLMs,  especially the increasing ability of the in-context learning with more exemplaries  with true labels.

In the following, we will use Eq. (\ref{eq:GD}) to formulate our following framework for ICL with local differential privacy. 

\subsection{Locally Differentially Private ICL}

Here we describe a threat model and emphasize the importance of local differential privacy in the preservation of individual privacy.  In this model, an organization owns a fully private database for some specific task (for example, presidential voting data, the school students' health records) and  hosts large language models (LLMs) via an API
endpoint, allowing users to query the LLM for answers based on the private data.  Sometimes, we assume that LLMs are frozen without any update of parameters \citep{Panda2023differentially,Duan2023flocks}.  We know that, in this scenario,  privacy leakage occurs under a canonical private attack called membership inference attack (MIA) which assesses whether a data point is used in the prompts appended to the inputs of a trained LLM \citep{Duan2023flocks}. The above formula Eq. (\ref{eq:GD}) explains well the membership inference attack. Given a query $q=(\boldsymbol{x_{\text{test}}}, \underline{?})$, an adversary tries to  distinguish whether it is within an demonstration set $\mathcal{E}_n =\{(\boldsymbol{x_1}, y_1), \cdots, (\boldsymbol{x_n}, y_n)\}$.  Without loss of generality, we assume that $\boldsymbol{x_{\text{test}}} = \boldsymbol{x_1}\in \mathcal{X}_n$ where $\mathcal{X}_n = \{\boldsymbol{x}: (\boldsymbol{x},y)\in \mathcal{E}_n \text{ for some } y \}$, i.e., $\boldsymbol{x_{\text{test}}}$ is within the demonstration set $\mathcal{E}_n$.  Let $\mathcal{E}_n':= \mathcal{E}_n \setminus \{(\boldsymbol{x_1}, y_1) \}\cup \{(\boldsymbol{x}', y')\}$ where $(\boldsymbol{x}',y') \not \in \mathcal{E}_n$. 
In particular, Since $\boldsymbol{x_{\text{test}}} \neq \boldsymbol{x}'$, the similarity $\boldsymbol{x_{\text{test}}^T} \boldsymbol{x_{\text{test}}} $ is usually much bigger than the similarity $\boldsymbol{x_{\text{test}}^T} \boldsymbol{x}'$. It follows from Eq. (\ref{eq:GD}) that $P(y_{\text{test}} | \mathcal{E}_n, \boldsymbol{x_{\text{test}}}) > P(y_{\text{test}} | \mathcal{E}_n', \boldsymbol{x_{\text{test}}})$. In other words, the prediction probability of the answer $y_{\text{test}}$ to the  query $\boldsymbol{x_{\text{test}}}$ in the in-context learning with the demonstration set $\mathcal{E}_n$ should be larger than the prediction of the answer to the same query with the demonstration set $\mathcal{E}_n'$. So the adversary may easily use the query $\boldsymbol{x_{\text{test}}}$ to distinguish $\mathcal{E}_n$ and $\mathcal{E}_n'$ especially when $n$ is small, and hence distinguish between membership and non-membership.  However, when $n$ gets larger (say 32), the difference between these two probabilities $P(y_{\text{test}} | \mathcal{E}_n, \boldsymbol{x_{\text{test}}}) $ and $P(y_{\text{test}} | \mathcal{E}_n', \boldsymbol{x_{\text{test}}})$ is relatively very small. Then it is not easy to distinguish between membership and non-membership.

Moreover, LLMs are shown to memorize individual data from the original training data and to retain users' data from smaller private datasets used to fine-tune
them for downstream tasks \citep{Mireshghallah2022memorization,Zhang2021counterfactual,Ippolito2022preventing,Mccoy2023much}. If prompts contain sensitive information, the LLM might expose privacy during queries as in Samsung privacy leakage \citep{Mitchell2023Samsung,Duan2023flocks}.   In this case, we may regard LLMs as an \emph{untrusted} aggregator in in-context learning.

Our goal is to employ LLMs to provide accurate answers to different queries from users but protect the privacy of individual data in the database.  In this paper, we focus on the scenarios where the LLMs are untrusted in privacy and the \emph{labels} in in-context learning are sensitive information.  We employ local differential privacy for these settings.  Let $\mathcal{A}$ be the \emph{private} input alphabet set and $\mathcal{O}$ be a finite  output alphabet set.  

\begin{definition} A randomized mechanism $Q$ from $\mathcal{A}$ to $\mathcal{O}$ is called  \emph{$\epsilon$-locally differentially private} if, for any two inputs $a$ and $a'$, and any output event $O \subseteq \mathcal{O}$, the following inequalities holds
\begin{align}
       Q(O |a) \leq e^{\epsilon} Q(O | a') \label{def:LDP}
\end{align}
\end{definition}
It formulates the privacy requirement that, by observing the same outcome $O$, an adversary cannot reliably distinguish whether the conditioned input is $a$ or $a'$. So the privacy in the input alphabet is preserved. When the privacy index $\epsilon$ is smaller, it is more difficult for the adversary to tell the two inputs apart and hence the privacy-preserving is better.   Differential privacy (central or local) satisfies two important properties that are crucial for the practical uses of DP mechanisms. The first one is called \emph{composition}.  Multiple DP mechanisms can be adaptively composed and applied to the same dataset.  The second is called \emph{postprocessing}. If a mechanism is differentially private, then any post-processing applied to the output of that mechanism is also differentially private.

In particular,  we use the $k$-ary randomized response mechanism ($k$-RR for short) to protect the privacy with labels \citep{Kairouz2016discrete,Wang2017locally}. Let $\mathcal{Y} = \{y_1, \cdots, y_M\}$ be the label set  and $\mathcal{E}_n:=\{(\boldsymbol{x_1}, y_1), \cdots, (\boldsymbol{x_n}, y_n)\}$ be the set of demonstration examples where $y_1, \cdots, y_n \in \mathcal{Y}$. The $k$\emph{-ary randomized response} on the label set $\mathcal{Y}$ is a randomized mechanism which maps $\mathcal{Y}$ stochastically to itself as follows: 
\begin{equation}
Q_{k\text{-RR}}(y'| y) = \frac{1}{M-1+e^{\epsilon}}\left\{
\begin{array}{rl}
e^{\epsilon} & \text{if } y' =y,\\
1 & \text{if } y' \neq y.
\end{array} \right. \label{k-RR} \nonumber
\end{equation}
Note that here $k=M$. In $\mathcal{E}_n$, $y_i$ is considered to be the \emph{true} label of the input $\boldsymbol{x_i}$ $(1\leq i \leq n)$. 
If we apply the $k$-RR mechanism $Q_{k\text{-RR}}$ to protect the privacy in labels, we obtain the perturbed demonstration set $Q_{k\text{-RR}} (\mathcal{E}_n) = \{(\boldsymbol{x_1}, Q_{k\text{-RR}}(y_1)), \cdots, (\boldsymbol{x_n}, Q_{k\text{-RR}}(y_n)\}$ where $Q_{k\text{-RR}}(y_i) (1\leq i \leq n)$ is a perturbation of the true label $y_i$. 
When $k=2$, the $k$-RR mechanism is just the well-known Warner's mechanism. In this paper, since we mainly focus on the binary classification problem in in-context learning, we use Warner's mechanism to protect the privacy in labels. For simplcity, we denote $\mathcal{Y}= \{0,1\}$. The obfuscation from Warner's mechanism is illustrated in Figure \ref{fig:Obfuscation}.

\begin{figure}[!htbp]
    \centering
    \includegraphics[width=0.7\linewidth]{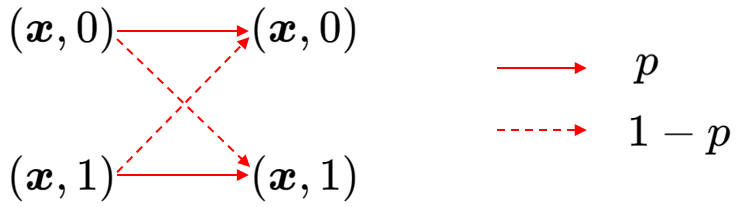}
    \caption{Obfuscation in labels, where $p=\frac{e^{\epsilon}}{e^{\epsilon}+1}$}.
    \label{fig:Obfuscation}
\end{figure}

Fix the demonstration set $\mathcal{E}_n:=\{(\boldsymbol{x_1}, y_1), \cdots, (\boldsymbol{x_n}, y_n)\}$. We apply the $\epsilon$-LDP $k\text{-RR}$ to $\mathcal{E}_n$ and obtain the perturbed demonstration set 
$Q_{k\text{-RR}}(\mathcal{E}_n)$ for the following ICL. 
 For any given query $\boldsymbol{x_{\text{test}}}$, we perform the ICL by querying the LLM GPT3.5 with the demonstration set $Q_{k\text{-RR}}(\mathcal{E}_n)$.  From the above analysis, we obtain the prediction probability of the true label $y_{\text{test}}$ for the query $\boldsymbol{x_{\text{test}}}$ as follows:
 \begin{small}
\begin{align}
   & P(y_{\text{test}} | Q_{k\text{-RR}} (\mathcal{E}_n), \boldsymbol{x_{\text{test}}}) \nonumber  \\  & = \sigma(\boldsymbol{W_0} \boldsymbol{x_{\text{test}}} - 
   \eta \sum_{i=1}^n ( \sigma(\boldsymbol{W_0 x_i})  
   - Q_{k\text{-RR}}(y_i)) \boldsymbol{x_i}^T \boldsymbol{x_{\text{test}}}) \label{eq:LDP-ICL}
\end{align}
\end{small}

The privacy of the labels is preserved from \emph{both} the untrusted LLMs and the observers of the query answers. By abuse of notion,  we use $Q^{(+\boldsymbol{x_{\text{test}}})}_{k\text{-RR}}(\mathcal{E}_n)$ for $\text{LLM}(Q_{k\text{-RR}}(\mathcal{E}_n), \boldsymbol{x_{\text{test}}})= \arg \max_{l\in \{0,1\}} P(\hat{y}_{\text{test}} = l | Q_{k\text{-RR}} (\mathcal{E}_n), \boldsymbol{x_{\text{test}}})$ to emphasize that the input of the randomized algorithm $Q^{(+\boldsymbol{x_{\text{test}}})}_{k\text{-RR}}$ is $\mathcal{E}_n$.  
The privacy of the true labels in the private set $\mathcal{E}_n$ is preserved by the randomized mechanism $Q_{k\text{-RR}}$. Indeed, given a prediction output probability for the label 1, the untrusted LLMs cannot reliably tell in the perturbed demonstration set $Q_{k\text{-RR}}(\mathcal{E}_n)$ which one of  0 and 1 is the true label of the input $\boldsymbol{x_i} (1\leq i \leq n)$.  For example,  although we may know that $(\boldsymbol{x}, 0)$ is in $Q_{k\text{-RR}}(\mathcal{E}_n)$, the LLM cannot be certain of the true label of $\boldsymbol{x}$ in the private demonstration set $\mathcal{E}_n$ because $(\boldsymbol{x}, 0)$ can be the output of both $(\boldsymbol{x},0)$ and $(\boldsymbol{x},1)$ under the randomized mechanism $Q_{k\text{-RR}}$ (Figure \ref{fig:Obfuscation}).  Let $\mathcal{E}_n^-$ be $\mathcal{E}_n \setminus \{(\boldsymbol{x_1}, y_1)\} \cup \{(\boldsymbol{x_1}, 1-y_1)\}$. In other words, $\mathcal{E}_n^-$ is obtained from $\mathcal{E}_n$ by flipping the label of the first element.  For any label $l\in \{0,1\}$, 
$ e^{-\epsilon} P( Q^{(+\boldsymbol{x_{\text{test}}})}_{k\text{-RR}}(\mathcal{E}_n) =l) \leq  P( Q^{(+\boldsymbol{x_{\text{test}}})}_{k\text{-RR}}(\mathcal{E}_n^-) =l) \leq e^{\epsilon} P( Q^{(+\boldsymbol{x_{\text{test}}})}_{k\text{-RR}}(\mathcal{E}_n) =l)$.  In this sense, the privacy in the true label is preserved by our $Q_{k\text{-RR}}$ against any observer of the query outcomes.

We propose LDP-ICL (Algorithm \ref{alg:ldp-icl-classification}), a new framework for protecting  private in-context learning demonstration examples. We randomly sample $n$ demonstration examples from the private dataset $\mathcal{D}$, perturb their labels, and then send these perturbed samples concatenated with the query to a LLM to predict the answer.  The algorithm is illustrated in Figure \ref{fig:main}. 

 \begin{algorithm}
 \caption{\textbf{LDP-ICL}}
 \label{alg:ldp-icl-classification}
 \begin{algorithmic}[1]
 \renewcommand{\algorithmicrequire}{\textbf{Input:}}
\renewcommand{\algorithmicensure}{\textbf{Output:}}
\Require  Private data $\mathcal{D}$, query $q$, model \textbf{LLM}, privacy budget $\epsilon$, number of demonstration examples $n$.
\Ensure Model prediction $O(q)$
    \State \textbf{Subsample} of size $n$ \textbf{from} $\mathcal{D}$ and obtain $\mathcal{E}_n$
    \State \textbf{Perturb} $\mathcal{E}_{n}$ using $k$-RR and obtain $Q_{k\text{-RR}}(\mathcal{E}_{n})$
    \State \textbf{Concatenate} query and form $I(q) = Q_{k\text{-RR}}(\mathcal{E}_{n}) \cup q$
    \State Obtain model output $O(q) = \textbf{LLM}(I(q))$
    \end{algorithmic}
    \end{algorithm}

There is a \emph{trade-off} between the privacy and utility in LDP-ICL, which is characterized by  the above formula (Eq. (\ref{eq:LDP-ICL})).  From this formula, we can theoretically compute the expected prediction probability of $Q_{k\text{-RR}}^{(+ \boldsymbol{x_{\text{test}}})} (\mathcal{E}_n)$ and its variance. In the private $\mathcal{E}_n$, when more (true) labels $y_i$ are flipped,  the terms $\sigma(\boldsymbol{W_0} \boldsymbol{x_i}) - y_i$ will flip the signs and hence the prediction output probabilities will decrease. When all the labels are flipped, the probability is the smallest. This also implies that, when $\epsilon$ in the randomization mechanism $Q_{k\text{-RR}}$ gets close to 0, the accuracy rate will decrease with the output probability in expectation.   We run some experiments on several datasets with different $\epsilon$. The results show a trade-off between the privacy index $\epsilon$ and the prediction accuracy.  If we want a better privacy-preservation for the sensitive labels, then $\epsilon$ must get smaller and hence the prediction accuracy decreases.  The results are illustrated in Figure \ref{fig:classification}. The privacy cost accumulates with more demonstrations and more queries according to the composition property of DP.



\subsection{Discrete Distribution Estimation}

Now we apply the above LDP-ICL to a touch-stone problem in DP: discrete distribution estimation problem. 
Assume that $\mathcal{D}$ is a given private database for some specific classification task (for example, the 2016 US Presidential Election Data) where each individual classification  label is sensitive.  Now we want to estimate the discrete distribution of different labels, i.e., the proportion of data points with each label in the whole database.  For simplicity and illustration, we choose the binary classification. Without loss of generality, let the label set be $\{0,1\}$, $\pi_0$ and $\pi_1$ be their unknown prior proportion in the database.  Now we select a finite set $\mathcal{D}_n$ of input-label pairs from the database $\mathcal{D}$ whose label distribution is the same as that of  the original dataset $\mathcal{D}$. Let $\mathcal{X}_{\mathcal{D}_n}: = \{\boldsymbol{x}: (\boldsymbol{x},y) \in \mathcal{D}_n \text{ for some } y\}$.     Now we use the above LDP-ICL to perturb the answer to each query from $\mathcal{X}_{\mathcal{D}_n}$.  We choose a demonstration set for each query and perturb the labels in the demonstration set  with LDP mechanism. With this noisy demonstration set for in-context learning, the answer of each query is also perturbed with a certain associated  probability without affecting the prediction accuracy much. 
The noisy answer is regarded as an privacy-preserving estimation of the true label. 
By collecting the noisy answers to all queries from $\mathcal{X}_{\mathcal{D}_n}$, we can estimate $\pi_1$ with local differential privacy. 

The crux of this approach is to choose the demonstration sets.  One possible solution is to choose \emph{a single} set of input-label pairs from $\mathcal{D}$  and its perturbed version as the demonstration set for \emph{all} queries from  $\mathcal{X}_{\mathcal{D}_n}$.  The problem with this approach is that each perturbed label in the demonstration set would expose to the untrusted LLM \emph{many times} (precisely $|\mathcal{D}_n|$ times) so that the adversary may estimate correctly the true labels in the demonstration set via LLM.  Our solution instead is to first partition $\mathcal{D}$ into different subsets of input-label pairs of relatively small size and try to pick up a different subset for each query.  The perturbed version of the subset is chosen as the demonstration set for the query.  In this way, we can avoid the possibility that a perturbed label might expose to LLM many times and hence we can estimate the proportion $\pi_1$ without leaking information about the true labels much.  Our approach is detailed in the following Algorithm \ref{algorithm:ldp-icl-estimation}.  Specifically, in line 1, we perform a random sampling of size $R$ from $\mathcal{D}$ without replacement to generate the query set. Notably, it is crucial to maintain a consistent proportion of  samples of labels 0 and 1 during the sampling process. In line 2, we split the original dataset into $l$ parts of demonstrations, each containing $n$ examples, where $l=|\mathcal{D}|/n$.   Finally, we obtain predicted answers for each query and calculate estimated positive rate
\begin{align}
    \widehat{P}_{t}  \triangleq \frac{\sum \limits_{i=1}^R\mathds{1}\left\{O(q^i)=1\right\}}{R} \label{eq.estimation}
\end{align}


 \begin{algorithm}
 \caption{\textbf{LDP-ICL for distribution estimation}}
 \label{algorithm:ldp-icl-estimation}
 \begin{algorithmic}[1]
 \renewcommand{\algorithmicrequire}{\textbf{Input:}}
\renewcommand{\algorithmicensure}{\textbf{Output:}}
\Require  Private data $\mathcal{D}$, model \textbf{LLM}, privacy budget $\epsilon$, number of demonstration examples $n$, number of round(queries) $R$
\Ensure Proportion estimation

\State \textbf{Subsample} of size $R$ \textbf{from} $\mathcal{D}$, obtain $\{(\boldsymbol{x^i_{\text{test}}},y^i_{\text{test}})\}_{i=1}^{R}$ and construct queries $\{q^i\}_{i=1}^{R} = \{(\boldsymbol{x^i_{\text{test}}},\underline{?})\}_{i=1}^{R}$
\State \textbf{Partition} $\mathcal{D}$ into classes with size $n$:  $\mathcal{D}^1_n, \dots, \mathcal{D}^l_n \leftarrow \mathcal{D}$
    \For { $i \in \{1,\dots,R\}$}
       \State \textbf{Perturb} $\mathcal{D}_{n}^{i}$ using $k$-RR and obtain $Q_{k\text{-RR}}(\mathcal{D}_{n}^{i})$
        \State \textbf{Concatenate} corresponding query and form $I(q^i) = Q_{k\text{-RR}}(\mathcal{D}_{n}^{i}) \cup q^i$
        \State Obtain $i$-th model output $O(q^i) = \textbf{LLM}(I(q^i))$
    \EndFor
\State Calculate estimated rate (Eq.(\ref{eq.estimation}))
    
    \end{algorithmic}
    \end{algorithm}

For each query $\boldsymbol{x_{\text{test}}^i}$, since $(\boldsymbol{x_{\text{test}}^i}, y_{\text{test}}^i) \in \mathcal{D}$, $y_{\text{test}}^i$ can be regarded as the true label of $\boldsymbol{x_{\text{test}}^i}$.  So $\text{LLM}(Q_{k\text{-RR}}(\mathcal{D}_{n}^{i}), \boldsymbol{x_{\text{test}}^i})$ is a perturbation of the true label $y_{\text{test}}^i$ according to the LDP-ICL.
Now we compare this in-context-learning approach with the well-known Warner's method on the same distribution estimation problem.    According to Warner's method, for each sample $(\boldsymbol{x_{\text{test}}^i}, y_{\text{test}}^i)$, we flip the true label $ y_{\text{test}}^i$ according to $Q_{\text{2-RR}}$ (Warner's mechanism). By collecting the noisy labels, we empirically estimate the proportion of the labels in the original private database $\mathcal{D}$.  The main difference between these two methods is that the LDP-ICL approach perturbs directly the labels in the \emph{demonstration set} but the Warner's method add noise directly to the \emph{queries}.  Generally, the LDP-ICL approach performs better in the high-privacy region ( when $\epsilon$ is relatively small) because the semantic prior of the LLM add some extra power to the estimation.  The experimental results are shown in Figure \ref{fig:estimation}. 



\section{Experiments} \label{sec:experiments}
We provide empirical results to demonstrate the effectiveness of our proposed LDP-ICL under two scenarios, classification(Alg.\ref{alg:ldp-icl-classification}) and distribution estimation(Alg.\ref{algorithm:ldp-icl-estimation}).

\subsection{Experimental Setup}

\subsubsection{Datasets and Model}
We evaluate LDP-ICL using four binary classification datasets, all obtained from Hugging Face. SST-2\citep{DBLP:conf/emnlp/SocherPWCMNP13} and Subj\citep{pang-lee-2004-sentimental} are for sentiment classification; Ethos\citep{mollas_ethos_2022} is a hate speech detection dataset; and SMS\_Spam\citep{misc_sms_spam_collection_228} is used for recognizing spam text messages. Dataset size details are provided in Table \ref{table:dataset}. The training set of each task is considered as private data. Test samples are randomly selected from the validation set in the classification scenario, while they are drawn from the training set in label distribution estimation scenario. We choose the GPT-3.5-turbo model for all tasks which demonstrates strong performance across various natural language tasks while offering a balanced combination of performance and cost-effectiveness. 
\begin{table}[!htbp]
\centering
\resizebox{0.5\textwidth}{!}
{
\begin{tabular}{|c|c|c|c|c|c|}
\hline
Task & Training dataset & Validation dataset  \\
\hline
SST-2   & 67349    & 872          \\
\hline
Subj    & 8000    & 2000   \\     
\hline
Ethos    & 500(998)    & 498(-)   \\     
\hline
SMS\_Spam    & 4070(5570)    & 1500(-)   \\     
\hline
\end{tabular}
}
\caption{The dataset size for each task. In the classification scenario, we partitioned the initial Ethos dataset (consisting of 998 training examples) and SMS\_Spam dataset (comprising 5570 training examples) into separate training and validation sets.}
\label{table:dataset}
\end{table}


\subsubsection{Baselines}
In the classification scenario, we compare LDP-ICL with the following baselines: 
\begin{itemize}
    \item \textbf{Non-private ICL}, i.e., $\epsilon= \infty$ in LDP-ICL, is equivalent to a non-private $n$-shot prediction.
    \item \textbf{Zero-Shot Learning (ZSL)} is the same as one-shot learning except that no demonstrations are allowed, and the model is only given a natural language instruction describing the task. ZSL performance enhances with larger model sizes, demonstrating commendable outcomes in GPT-3 175B\citep{DBLP:journals/corr/abs-2005-14165}.
    \item \textbf{Flipped-Label ICL(FL-ICL)} is a pattern that flips all class labels in the exemplars, indicating a disagreement between semantic prior knowledge and input-label mappings. The performance accuracy is \emph{inversely} proportional to the ability to learn input–label mappings and override semantic priors \citep{wei2023larger}.

\end{itemize}

For distribution estimation, we compare LDP-ICL with simple \textbf{Coin Flipping(CF)} (Warner's) method that randomly flips the true lables in test queries.

\subsubsection{Implementation Details}
In our experiments, we employ a uniform template for structuring examples across all tasks. Additionally, to mitigate the influence of sensitive settings on ICL performance, we ensure that our demonstration examples meet some constraints. Specific templates for each task and constraints are provided in the Appendix B and C. By default, we set the number of demonstrations to $n = 32$. 
For classifying unknown test queries, we first tune the random seed for each task to find a set of demonstration examples that achieves the best validation performance based on ICL. Then, the same ordered demonstration set is used for fair LDP-ICL comparisons across various discrete budget levels($\epsilon=\{0,0.5,1,2, 3, 8,\infty\}$). Finally, we utilize evaluation metric to assess model prediction accuracy on a subset of 150 test examples from the validation set. 
Evaluation performance is averaged over 6 runs under the same parameter configuration. In the distribution estimation scenario, we selected number of queries separately $R=1000$ for SST-2 dataset and $R=500$ for Ethos dataset.


\subsubsection{Evaluation Metrics}
In the classification scenario, we gauge our method's performance by measuring the accuracy between predicted answers with the true ones.
For estimating label distribution, we calculate estimated positive rate using Eq.(\ref{eq.estimation}).

\begin{figure*}[t]
    \centering
    \includegraphics[width=0.95\linewidth]{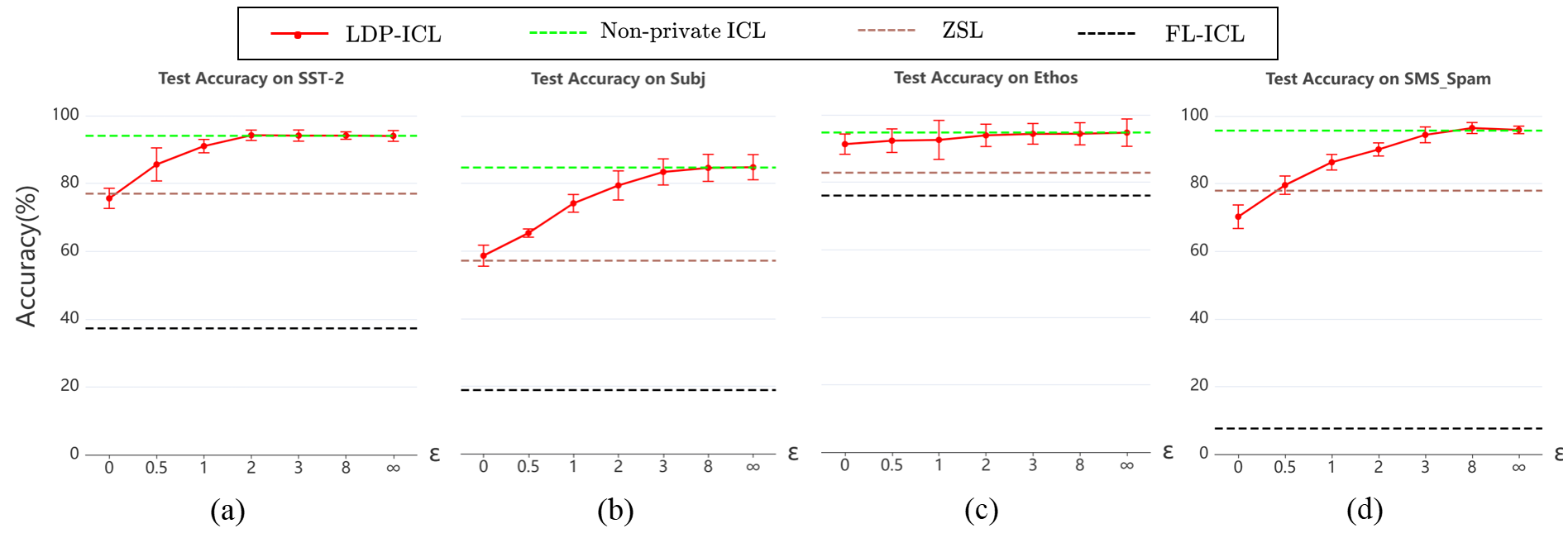}
    \caption{\textbf{Classification scenario:} Test performance on (a)SST-2, (b)Subj, (c)Ethos and (d) SMS\_Spam}.
    \label{fig:classification}
\end{figure*}

\subsection{Experimental Results} 
This part consists of LDP-ICLs for classification and distribution estimation. 

\subsubsection{LDP-ICL for classification}
Figure \ref{fig:classification} shows the performance of our LDP-ICL as well as three baselines in all tasks. As can be seen, the consistently lowest accuracy of the baseline FL-ICL method, falling below 50\% for all datasets(except Ethos), suggests vulnerability to perturbed class labels. This vulnerability is attributed to GPT-3.5-turbo's emergent task learning capability, enabling it to learn input-label mapping that override established semantic priors\citep{wei2023larger}.
Comparing against two additional baselines, namely the lower bound ZSL and the upper bound  ICL, LDP-ICL demonstrates significant enhancements over ZSL and achieves competitive results similar to non-private ICL when $\epsilon \geq 3$. This observation underscores the beneficial impact of the optimization performed by LDP-ICL on downstream tasks.
Furthermore, it's worth highlighting that with $\epsilon\geq8$, the privacy protection becomes almost negligible, leading to indistinguishable performance between LDP-ICL and the non-private ICL setting. Overall, reducing budgets $\epsilon$ strengthens privacy assurance in LDP-ICL but inevitably hampers downstream task performance. Specifically, at $\epsilon = 0$, half of the demonstration example class labels are inverted, yielding performance on par with ZSL in expectation. Conversely, at $\epsilon = \infty$, there is no privacy safeguard and our LDP-ICL degrades to ICL. A more detailed analysis of those tasks reveals that to attain or approach non-private ICL performance, a slightly different budget value is needed. This implies that downstream manufacturers should select the appropriate privacy protection parameter $\epsilon$ based on task-specific needs without losing much utility.

\subsubsection{LDP-ICL for distribution estimation}
Figure \ref{fig:estimation} presents a comparison between the performance of LDP-ICL and CF. The results reveal that our LDP-ICL estimation aligns more closely with the true proportion and maintains a higher level of stability especially in cases of smaller $\epsilon$, demonstrating better utility. Since we typically prefer smaller budget, which indicates stronger privacy, LDP-ICL outperforms CF in terms of utility.



 \begin{figure}[!htbp]
    \centering
    \includegraphics[width=\linewidth]{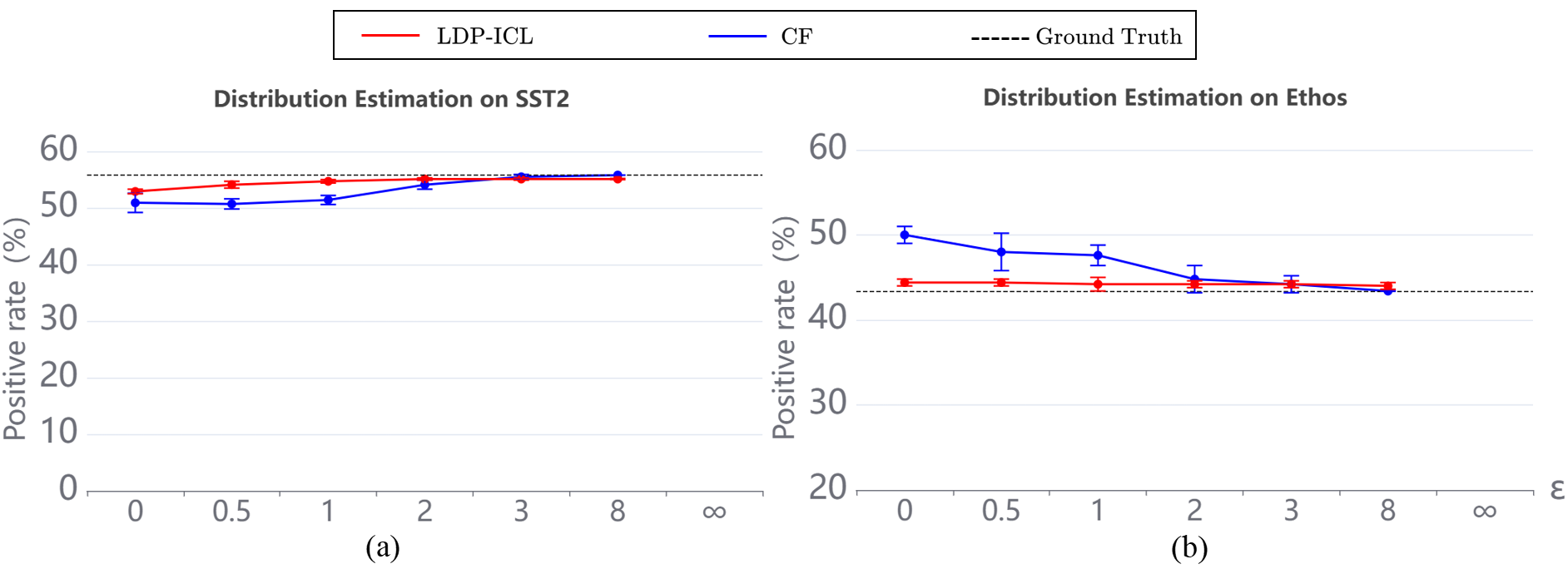}
    \caption{\textbf{Distribution estimation scenario:} Estimation results on (a)SST-2 and (b)Ethos.}.
    \label{fig:estimation}
\end{figure}
 Our initial
analysis indicates that for a given privacy parameter value, a higher quantity of examples leads to
a higher count of flipped examples, which implies a more powerful task-learning ability and hence
a less accurate prediction rate.

\subsubsection{Comparison Experiments}

We have performed comparison experiments with other three representative privacy-preserving methods: DP-SGD, DP-ICL and PromptPATE and their comparison results are listed in Table \ref{tab:performance_comparison}.  The results have demonstrated that locally differentially private ICL also can reach the utility level of other privacy-preserving methods. 


\begin{table}[]
\centering
\resizebox{0.5\textwidth}{!}
    {
\begin{tabular}{|c|c|c|c|c|}
\hline
Model                          & Method        & $\epsilon=3$    & $\epsilon=8$     & $\epsilon=\infty$  \\ \hline
\multirow{3}{*}{\begin{tabular}[c]{@{}l@{}}
             \enspace \thinspace RoBERTa\\ \quad \ \  -large
             \end{tabular} } & DP-SGD        & $93.04$ & $93.81$ & $96.2 $ \\ \cline{2-5} 
                               & DP-SGD        & $94.6$  & $94.7 $ & $95.5 $ \\ \cline{2-5} 
                               & DP-SGD        & $94.23$ & $94.87$ & $96.2$  \\ \hline
\begin{tabular}[c]{@{}l@{}}
             \enspace \thinspace RoBERTa\\ \quad \  -base
             \end{tabular}                   & promptPATE    & $86.35$ & $92.32$ & $93.23$ \\ \hline
\multirow{2}{*}{\begin{tabular}[c]{@{}l@{}}
             \enspace \thinspace GPT-3\\ Babbage
             \end{tabular} }  & DP-ICL($n=4$)   & $95.8$  & $95.92$ & $96.05$ \\ \cline{2-5} 
                               & DP-ICL($n=16$)  & $91.64$ & $96.32$ & $96.13$ \\ \hline
\multirow{2}{*}{\begin{tabular}[c]{@{}l@{}}
             \enspace \thinspace GPT-3.5\\  \quad \  Turbo
             \end{tabular} } & LDP-ICL($n=16$) & $94.45$ & $94.9$  & $95.77$ \\ \cline{2-5} 
                               & LDP-ICL($n=32$) & $94.11$ & $94.12$ & $94.12$ \\ \hline
\end{tabular}
}
\caption{Performance comparison of DP-SGD, promptPATE, DP-ICL and LDP-ICL under various privacy budgets.}
\label{tab:performance_comparison}
\end{table}


\subsubsection{Ablation study}
Our intuition for choosing demonstration examples was to assess whether a model can learn input-label mappings and override semantic priors. A performance below 50\% accuracy in FL-ICL indicates the model's ability to achieve this\cite{wei2023larger}. 

\begin{table}[!htbp]
\centering
\begin{tabular}{|c|ccccc|}
\hline
\multirow{2}{*}{Dataset} & \multicolumn{5}{c|}{Number of demonstration examples}                                                              \\ \cline{2-6} 
                         & \multicolumn{1}{c|}{4}  & \multicolumn{1}{c|}{8}  & \multicolumn{1}{c|}{16} & \multicolumn{1}{c|}{32} & 64 \\ \hline
SST-2                    & \multicolumn{1}{c|}{55} & \multicolumn{1}{c|}{44} & \multicolumn{1}{c|}{33} & \multicolumn{1}{c|}{31} & 32 \\ \hline
Subj                     & \multicolumn{1}{c|}{93} & \multicolumn{1}{c|}{78} & \multicolumn{1}{c|}{64} & \multicolumn{1}{c|}{38} & 39 \\ \hline
\end{tabular}
\caption{Performance of FL-ICL over number of demonstration examples on SST-2 and Subj.}
\label{Table:FL-ICL}
\end{table}

Table \ref{Table:FL-ICL} presents the performance for the selection of  $n=32$ demonstration examples, which indicates that the accuracy rate falls below 50\% and hence show  LLM's capability to learn input-label mappings and override semantic priors. Additionally, we carried out ablation studies to analyze how varying the quantity of demonstration examples affects the sentiment analysis performance for the SST-2 task. The analysis was conducted under three demonstration number cases: $n = 16, 32, 64$.

As depicted in Figure \ref{fig:ablation}, we find that these three curves exhibit an identical trend of change regardless of the variation in example quantities. A  variation is observed when the privacy parameter falls within the range of 0 to 0.5: with an increase in the number of examples, accuracy diminishes. Our initial analysis indicates that for a given privacy parameter value, a higher quantity of examples leads to a higher count of flipped examples, which implies a more powerful task-learning ability and hence a less accurate prediction rate. Drawing on the preceding formula (Eq. (5)) in the main text), the heightened presence of flipped examples corresponds to a more pronounced influence on the accuracy.
 \begin{figure}[!htbp]
    \centering
    \includegraphics[width=0.7\linewidth]{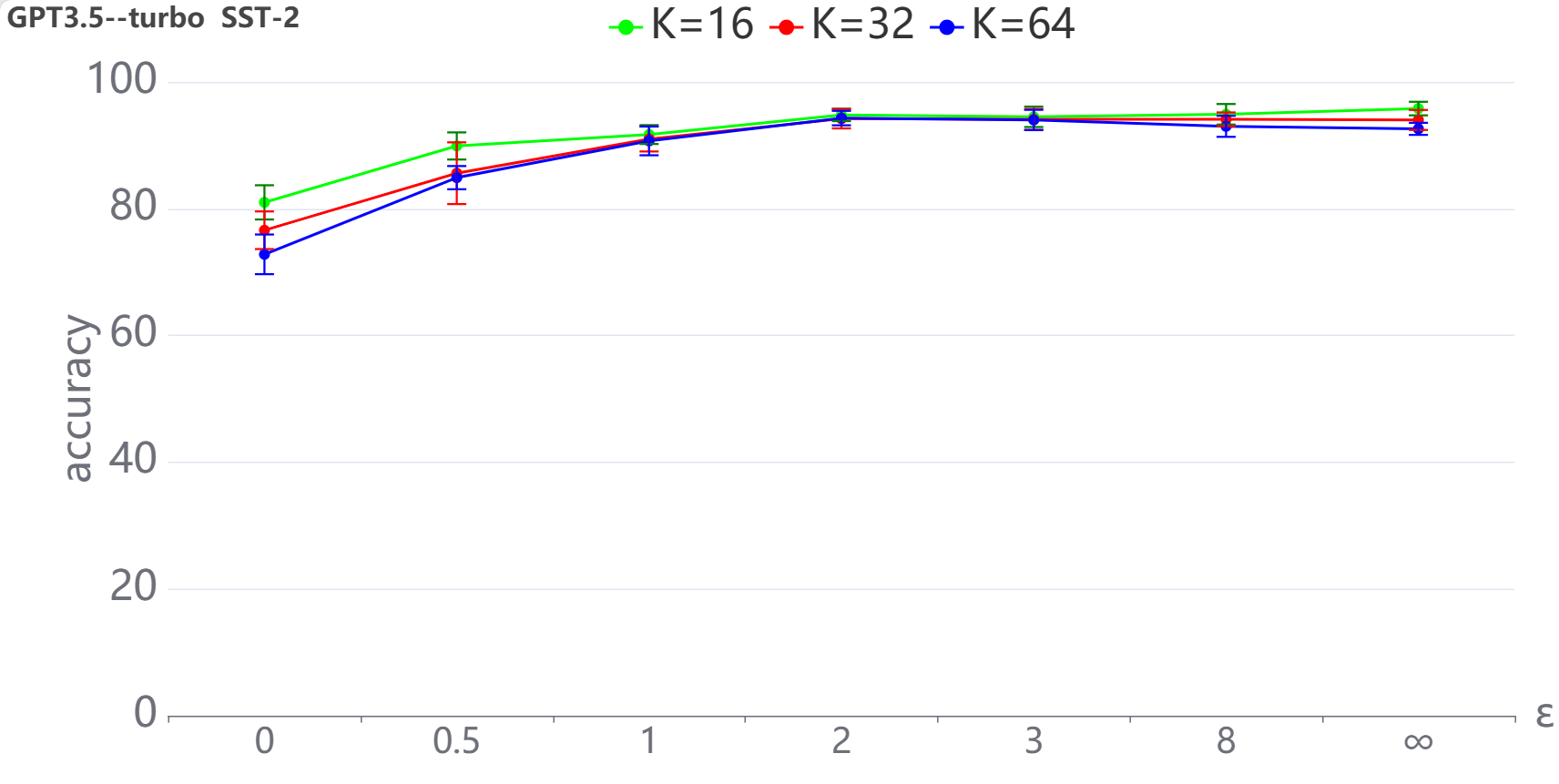}
    \caption{Performance across numbers of the examples}
    \label{fig:ablation}
\end{figure}

\section{\textbf{Related Works and Conclusion}} \label{sec:conclusions}

In this paper, following the tradition in DP \citep{Dinur2003revealing} literature, we treat the inputs in the input-label pairs as identifiers and hence nonsensitive but regard \emph{only the labels as sensitive}.  In this setting, 
we are \emph{the first} to study the locally differentially private ICL. Our privacy-preserving of labels is different from the so-called label differential privacy \cite{Ghazi2021deep}, which is essentially central DP. There is a rich literature on privacy-preserving ICL. Here we only discuss some closely related to our work.  There are some works which work on the privacy-preservation for ICL but mainly focus on the central DP that assume that the curator is trusted \citep{Duan2023flocks,Panda2023differentially}.  In some sense, our approach is similar to PromptPATE in \citep{Duan2023flocks} that both use noisy prompts prepended to a query to perform in-context learning. The main difference is the method to add noises. In PromptPATE, the noise is added to the ensembled result in a central way while in our LDP-ICL, we add noise locally to the labels in prompts.  In \citep{Yu2021differentially,Li2021large}, they  deal with the DP fine tuning of the parameters of the LLMs.  In contrast, we regard LLMs frozen and the in-context learning proceeds without modifications of the parameters.  In \citep{Li2023privacy}, they use text-to-text privatization while our work focuses on only the privatization of the labels.  We adapt the ideas of ICL as Transformer attention mechanism by a dual implicit gradient descent optimization from \citep{Dai2022can,von2023transformers,Irie2022dual}. But those papers mainly deal with linear regression problem while we work on the classification problem. Proposition 1 in our paper is based on Appendix A in \citep{Zhmoginov2022hypertransformer}.   It this paper, our experiments are run on some common datasets for classification which are not privacy-sensitive (probably Ethos is an exception). In the future, we will try a privacy-sensitive synthetic dataset.  In this paper, the perturbation is on the labels only, which is a quite limited case. We plan to employ LDP for more general cases of demonstrations in ICL. 

The selection of demonstrations is an important issue \citep{Zhang2022active,Rubin2022learning,Zhao2021calibrate,Dong2022survey} that we have not addressed yet in this paper.  From our formula (Eq. (\ref{eq:LDP-ICL})), we know that a good selection of demonstrations can improve the trade-off between privacy and utility of LDP-ICL. We would like to find an optimal adaptive selection algorithm for  our LDP-ICL. In this paper, we treat only labels as sensitive. We also will privatize the input sentences or words with local differential privacy \cite{Du2023sanitizing,Li2023privacy,Yue2021differential}.


\section{References} \label{sec:reference}
\bibliographystyle{lrec-coling2024-natbib}
\bibliography{Savage}

\clearpage

\renewcommand\thesection{\Alph{section}}
\setcounter{section}{0}

\section*{\centering \LARGE \textbf{Appendix}}

\section{Proof of Proposition 1}

For simplicity, we don't specify the sizes of different matrices. The context will determine the sizes.  

We define matrix $P$ and  operator $\sigma^-$  by 

\[ P = \eta I, \text{ and }
\sigma^-(  
\left[
  \begin{array}{c}
    \boldsymbol{W}_0 x_i   \\
     y_i  \\
 \end{array}
\right] )
=
\begin{bmatrix}
          0 \\
          \sigma(\boldsymbol{W}_0 x_i)-y_i
         \end{bmatrix}
\]

where $I$ is the identity matrix. Note that $\sigma^-$ is just a sigmoid function followed by a subtraction. Define
\[
\boldsymbol{W}_K = \boldsymbol{W}_Q =
\left[
  \begin{array}{cc}
     I & 0   \\
    0 & 0   \\
 \end{array}
\right],
\boldsymbol{W}_V= 
\begin{bmatrix}
         \boldsymbol{W}_0 & 0 \\
          0 & I 
         \end{bmatrix}
\]

Consider

\begin{align*}
     P
\sum_{i=1}^n [&(\sigma^-(
\left[
  \begin{array}{cc}
     \boldsymbol{W}_0 &  0\\
   0  &  I\\
     \end{array}
\right]
\begin{bmatrix}
           x_i \\
          y_i 
         \end{bmatrix}) \\
& \otimes 
(\begin{bmatrix}
           I & 0 \\
          0 & 0 
         \end{bmatrix}
         \begin{bmatrix}
           x_i \\
          y_i 
         \end{bmatrix})]
    (\begin{bmatrix}
           I & 0 \\
          0 & 0 
         \end{bmatrix}
         \begin{bmatrix}
           x_{\text{test}} \\
           p_{\text{test}} 
         \end{bmatrix})
 \end{align*}

We can compute the above expression and obtain that it is equal to $\eta \sum_{i=1}^n (\sigma(\boldsymbol{W}_0 x_i)- y_i)) \boldsymbol{x_i}^T \boldsymbol{x}_{\text{test}}$, which is just $-\Delta \boldsymbol{W} x_{\text{test}}$.  Then the matrices $\boldsymbol{W}_{\boldsymbol{v}}, \boldsymbol{W}_{\boldsymbol{k}}$ and $\boldsymbol{W}_{\boldsymbol{q}}$ can be constructed from $\sigma^-$ and the above matrices $\boldsymbol{W}_V$, $\boldsymbol{W}_K$ and $\boldsymbol{W}_Q$  respectively.

\section{Templates for In-context Learning}
\subsection{Format}
In this paper, we adopt a uniform prompt template for each task, which is structured to present inputs and their corresponding labels. We use a newline to separate the input and the label as well as each demonstration example. All prompts are designed to be answered with a single label responses (e.g., "Positive/Negative", "Subject/Object") so that we can directly check the prediction results of LLM instead of applying decoding methods.

\begin{table*}[!htbp]
\centering
\begin{tabular}{c|c|c}
\hline
Dataset   & Template                                                                        & Label Set                       \\ \hline
SST-2     & \begin{tabular}[c]{@{}l@{}} Input: \{Sentence\}\\ Output: \{Label\}\end{tabular} & \{Negative,Positive\}           \\ \hline
Subj      & \begin{tabular}[c]{@{}l@{}}Input: \{Sentence\}\\ Output: \{Label\}\end{tabular} & \{Objective,Subjective \}       \\ \hline
Ethos     & \begin{tabular}[c]{@{}l@{}}Input: \{Comment\}\\ Output: \{Label\}\end{tabular}  & \{Not hate speech,Hate speech\} \\ \hline
SMS\_Spam & \begin{tabular}[c]{@{}l@{}}Input: \{Message\}\\ Output: \{Label\}\end{tabular}  & \{Ham,Spam\}                    \\ \hline
\end{tabular}

\label{Table:table1}

\caption{Templates and corresponding label set format for four classification tasks.}
\end{table*}

\subsection{Prompt examples}

\subsubsection{SST-2}
\setlength{\leftmargini}{2em} 
\begin{itemize}
    



\item \textbf{Input}: hide new secretions from the parental units 

\textbf{Output}: Negative

\item \textbf{Input}: that loves its characters and communicates something rather beautiful about human nature 

\textbf{Output}: Positive

\item \textbf{Input}: contains no wit , only labored gags 

\textbf{Output}: Negative

\item \textbf{Input}: demonstrates that the director of such hollywood blockbusters as patriot games can still turn out a small , personal film with an emotional wallop . 

\textbf{Output}: Positive

\item \textbf{Input}: remains utterly satisfied to remain the same throughout 

\textbf{Output}: Negative

\item \textbf{Input}: of saucy 

\textbf{Output}: Positive

\item \textbf{Input}: on the worst revenge-of-the-nerds clichs the filmmakers could dredge up 

\textbf{Output}: Negative

\item \textbf{Input}: are more deeply thought through than in most ` right-thinking' films 

\textbf{Output}: Positive

\item \textbf{Input}: that's far too tragic to merit such superficial treatment 

\textbf{Output}: Negative

\item \textbf{Input}: the greatest musicians

\textbf{Output}: Positive

\item \textbf{Input}: a depressed fifteen-year-old's suicidal poetry 

\textbf{Output}: Negative

\item \textbf{Input}: with his usual intelligence and subtlety 

\textbf{Output}: Positive

\item \textbf{Input}: goes to absurd lengths 

\textbf{Output}: Negative

\item \textbf{Input}: swimming is above all about a young woman's face , and by casting an actress whose face projects that woman's doubts and yearnings , it succeeds . 

\textbf{Output}: Positive

\item \textbf{Input}: for those moviegoers who complain that ` they don't make movies like they used to anymore 

\textbf{Output}: Negative

\item \textbf{Input}: equals the original and in some ways even betters it 

\textbf{Output}: Positive

\item \textbf{Input}: the part where nothing's happening , 

\textbf{Output}: Negative

\item \textbf{Input}: if anything , see it for karen black , who camps up a storm as a fringe feminist conspiracy theorist named dirty dick . 

\textbf{Output}: Positive

\item \textbf{Input}: saw how bad this movie was 

\textbf{Output}: Negative

\item \textbf{Input}: a smile on your face 

\textbf{Output}: Positive

\item \textbf{Input}: lend some dignity to a dumb story 

\textbf{Output}: Negative

\item \textbf{Input}: comes from the brave , uninhibited performances 

\textbf{Output}: Positive

\item \textbf{Input}: cold movie 

\textbf{Output}: Negative

\item \textbf{Input}: enriched by an imaginatively mixed cast of antic spirits 

\textbf{Output}: Positive

\item \textbf{Input}: redundant concept 

\textbf{Output}: Negative

\item \textbf{Input}: in world cinema 

\textbf{Output}: Positive

\item \textbf{Input}: excruciatingly unfunny and pitifully unromantic 

\textbf{Output}: Negative

\item \textbf{Input}: very good viewing alternative 

\textbf{Output}: Positive

\item \textbf{Input}: which half of dragonfly is worse : the part where nothing's happening , or the part where something's happening 

\textbf{Output}: Negative

\item \textbf{Input}: on all cylinders 

\textbf{Output}: Positive

\item \textbf{Input}: the plot is nothing but boilerplate clichs from start to finish , 

\textbf{Output}: Negative

\item \textbf{Input}: more than another `` best man '' clone by weaving a theme throughout this funny film 

\textbf{Output}: Positive

\item \textbf{Input}: the action is stilted 

\textbf{Output}: 

Answer: Negative

\end{itemize}

\subsubsection{Subj}
\begin{itemize}
    





\item \textbf{Input}: the tucks have a secret , they're immortal . they

\textbf{Output}: Objective

\item \textbf{Input}: check your brain and your secret agent decoder ring at the door because you don't want to think too much about what's going on . the movie does has some entertainment value - how much depends on how well you like chrisrock .

\textbf{Output}: Subjective

\item \textbf{Input}: this could be lizzy's only chance to start a new life and recreate the family she tragically lost as a child .

\textbf{Output}: Objective
\item \textbf{Input}: a beautifully tooled action thriller about love and terrorism in korea .

\textbf{Output}: Subjective
\item \textbf{Input}: the book tells of murray , the old scot patriot , who has had his eyes torn out and his house taken away during the english invasion .

\textbf{Output}: Objective
\item \textbf{Input}: presents a most persuasive vision of hell on earth .

\textbf{Output}: Subjective
\item \textbf{Input}: naturally , he returns to his analyst dr . ben sobel ( crystal ) for help and finds that sobel needs some serious help himself as he has inherited the family practice , as well as an excess stock of stress .

\textbf{Output}: Objective
\item \textbf{Input}: the film's sharp , often mischievous sense of humor will catch some off guard . . .

\textbf{Output}: Subjective
\item \textbf{Input}: still suffering from her hangover , julie does n't realize that ellen is missing when the school bus leaves the cemetery .

\textbf{Output}: Objective
\item \textbf{Input}: buries an interesting storyline about morality and the choices we make underneath such a mountain of clichM-is and borrowed images that it might more accurately be titled mr . chips off the old block .

\textbf{Output}: Subjective
\item \textbf{Input}: several people are listening to keith's plight on the radio and are making changes of their own .

\textbf{Output}: Objective
\item \textbf{Input}: like the hugely successful first film , the sequel provides pleasant , engaging if largely disposable entertainment for the general audience .

\textbf{Output}: Subjective
\item \textbf{Input}: real escape , however , may be a goal beyond possibility .

\textbf{Output}: Objective
\item \textbf{Input}: a sweet , engaging family story about the nature of faith and the power of youthful innocence .

\textbf{Output}: Subjective
\item \textbf{Input}: unfortunately , he neglects to mention that he is a professional gambler with bad instincts and that the jaguar is borrowed .

\textbf{Output}: Objective
\item \textbf{Input}: the filmmakers skillfully evoke the sense of menace that nature holds for many urban dwellers .

\textbf{Output}: Subjective
\item \textbf{Input}: in this racially-charged climate , the lapd's elite special investigations squad ( sis ) is assigned a high-profile quadruple homicide .

\textbf{Output}: Objective
\item \textbf{Input}: a remarkable 179-minute meditation on the nature of revolution .

\textbf{Output}: Subjective
\item \textbf{Input}: this documentary follows a business partnership between a clan of palestinian fisherman from refugee camps in gaza and israelis from the settlement of dugit , gaza strip .

\textbf{Output}: Objective
\item \textbf{Input}: . . . routine , harmless diversion and little else .

\textbf{Output}: Subjective
\item \textbf{Input}: this painter's block has lead to a desperate man obsessed with the need to paint , leaving the artist in a state of disarray .

\textbf{Output}: Objective
\item \textbf{Input}: i suspect this is the kind of production that would have been funnier if the director had released the outtakes theatrically and used the film as a bonus feature on the dvd .

\textbf{Output}: Subjective
\item \textbf{Input}: all social structures break down and a new world order emerges from the heart of the desert .

\textbf{Output}: Objective
\item \textbf{Input}: fresh and raw like a blown-out vein , narc takes a walking-dead , cop-flick subgenre and beats new life into it .

\textbf{Output}: Subjective
\item \textbf{Input}: when blake is initiated , he agrees to abide by the three cardinal rules : share your spirit , body and soul with others ; devote your life to saving homeless kids ; and never leave unless the master casts you out .

\textbf{Output}: Objective
\item \textbf{Input}: its aims are modest -- to make kids laugh as cuba gooding falls down a lot and to show how thrilling sled dogs are -- and it achieves them .

\textbf{Output}: Subjective
\item \textbf{Input}: desert rose is an erotic thriller that takes you through one man's life experience .

\textbf{Output}: Objective
\item \textbf{Input}: on the surface , the film is played straight , with sustained tension and effective frights . but the director has done his research , and fills every scene with a subversive humor that's 100 \% pure imitation hitchcock .

\textbf{Output}: Subjective
\item \textbf{Input}: witt plays a young harvard graduate who also represents grant's character .

\textbf{Output}: Objective
\item \textbf{Input}: . . . digs beyond the usual portrayals of good kids and bad seeds to reveal a more ambivalent set of characters and motivations .

\textbf{Output}: Subjective
\item \textbf{Input}: blade's team and the bloodpack are the only line of defense which can prevent the reaper population from wiping out the vampire and human populations .

\textbf{Output}: Objective
\item \textbf{Input}: if the count of monte cristo does n't transform caviezel into a movie star , then the game is even more rigged than it was two centuries ago .

\textbf{Output}: Subjective
\item \textbf{Input}: they intend to bring her back to tel-aviv in order to save her life .

\textbf{Output}: 

Answer: Objective

\end{itemize}

\subsubsection{Ethos}
\begin{itemize}
    



\item \textbf{Input}: men and women are not equal.  irrational contrary belief and policy only result in mounting failure.

\textbf{Output}: Not hate speech
\item \textbf{Input}: Asians have more fats so they are yellow.  So fat people are yellow?...

\textbf{Output}: Hate speech
\item \textbf{Input}: I'm very sad about this fucking situation with all this muslims

\textbf{Output}: Not hate speech
\item \textbf{Input}: As an SJW, it is the right of black people to rob and kill me. To not allow them is a white privilege.

\textbf{Output}: Hate speech
\item \textbf{Input}: You look like someone who would do an electric wheelchair race with Stephen Hawking.

\textbf{Output}: Not hate speech
\item \textbf{Input}: Any lesbian that doesn't look like a complete hog actually wants to be with a man.

\textbf{Output}: Hate speech
\item \textbf{Input}: You can ask me to call you a male, but you'll always be a woman

\textbf{Output}: Not hate speech
\item \textbf{Input}: An ugly man in a wig and dress, is still just an ugly man in a wig and dress... NOT A WOMAN!!!

\textbf{Output}: Hate speech
\item \textbf{Input}: Yo, if I’m going to hell for being homosexual, why did god make me like this?

\textbf{Output}: Not hate speech
\item \textbf{Input}: an Asian Moe Howard piece of shit

\textbf{Output}: Hate speech
\item \textbf{Input}: Why are people so dumb? The police were yelling at everyone to get out of the area yet people just stood there or recorded, guess they want to die.

\textbf{Output}: Not hate speech
\item \textbf{Input}: American people are dumbass. They don't even know geography

\textbf{Output}: Hate speech
\item \textbf{Input}: White police officers are always the victims.

\textbf{Output}: Not hate speech
\item \textbf{Input}: American Indians are finally almost extinguished

\textbf{Output}: Hate speech
\item \textbf{Input}: WHAT THE FUCK IS THIS ??!????! Disgusting! Absolutely disgusted and disappointed in Today's America.

\textbf{Output}: Not hate speech
\item \textbf{Input}: All muslims worship Satan¬ü¬ï¬ã 5 times a day

\textbf{Output}: Hate speech
\item \textbf{Input}: Well..congratulations u loose one of ur subscriber.... sHame on u guys u guys sick.

\textbf{Output}: Not hate speech
\item \textbf{Input}: all mental illnesses are awful and must be treated

\textbf{Output}: Hate speech
\item \textbf{Input}: We need to exterminate people that Subject children to such Perversion

\textbf{Output}: Not hate speech
\item \textbf{Input}: All girls like you should have a pimp to guide them to successful prostitution

\textbf{Output}: Hate speech
\item \textbf{Input}: We are repeating our past before it used to Germany and Jews now its China and Muslims

\textbf{Output}: Not hate speech
\item \textbf{Input}: After all this time the african americans are the people with most crime rate

\textbf{Output}: Hate speech
\item \textbf{Input}: Trumps a liar and an idiot, a douche bag , Narcissist, and Zionist boot licker.

\textbf{Output}: Not hate speech
\item \textbf{Input}: actually autism can make you a retard

\textbf{Output}: Hate speech
\item \textbf{Input}: Those illegals want to sue the company for money owed.What about the money they owe America. It cost money to educated, medical care, welfare, ect.

\textbf{Output}: Not hate speech
\item \textbf{Input}: Jews are lower class pigs.

\textbf{Output}: Hate speech
\item \textbf{Input}: this women has to be some sort of psychopath but it goes both ways, if either a male or a female are in a toxic relationship they need to get help, the sooner the better.

\textbf{Output}: Not hate speech
\item \textbf{Input}: Send all these whiny black crybabies back to central Apefrica to live for just one month. They would be crying to come back to America where White people pamper and coddle them. Let's try the experiment one time.

\textbf{Output}: Hate speech
\item \textbf{Input}: This people are treating women like dogs!!!!!!

\textbf{Output}: Not hate speech
\item \textbf{Input}: Trans people=gays and loosers

\textbf{Output}: Hate speech
\item \textbf{Input}: This fuckhead makes me asHamed of being Jewish. Shit people like him are destroying everything

\textbf{Output}: Not hate speech
\item \textbf{Input}: The Jew Faggot Behind The Financial Collapse.

\textbf{Output}: Hate speech
\item \textbf{Input}: This cop had to get him under arrest for walking like that freely out in public. Problematics must be stay home locked. Sue her parents.

\textbf{Output}: 

Answer: Not hate speech

\end{itemize}

\subsubsection{SMS\_Spam}
\begin{itemize}
    





\item \textbf{Input}: The wine is flowing and i'm i have nevering.

\textbf{Output}: Ham
\item \textbf{Input}: Customer service annoncement. You have a New Years delivery waiting for you. Please call 07046744435 now to arrange delivery

\textbf{Output}: Spam
\item \textbf{Input}: Yup i thk cine is better cos no need 2 go down 2 plaza mah.

\textbf{Output}: Ham
\item \textbf{Input}: You are a winner U have been specially selected 2 receive £1000 cash or a 4*holiday (flights inc) speak to a live operator 2 claim 0871277810810

\textbf{Output}: Spam
\item \textbf{Input}: Ok... Ur typical reply...

\textbf{Output}: Ham
\item \textbf{Input}: -PLS STOP bootydelious (32/F) is inviting you to be her friend. Reply YES-434 or NO-434 See her: www.SMS.ac/u/bootydelious STOP? Send STOP FRND to 62468

\textbf{Output}: Spam
\item \textbf{Input}: As per your request 'Melle Melle (Oru Minnaminunginte Nurungu Vettam)' has been set as your callertune for all Callers. Press *9 to copy your friends Callertune

\textbf{Output}: Ham
\item \textbf{Input}: BangBabes Ur order is on the way. U SHOULD receive a Service Msg 2 download UR content. If U do not, GoTo wap. bangb. tv on UR mobile internet/service menu

\textbf{Output}: Spam
\item \textbf{Input}: You are everywhere dirt, on the floor, the windows, even on my shirt. And sometimes when i open my mouth, you are all that comes flowing out. I dream of my world without you, then half my chores are out too. A time of joy for me, lots of tv shows i.ll see. But i guess like all things you just must exist, like rain, hail and mist, and when my time here is done, you and i become one.

\textbf{Output}: Ham
\item \textbf{Input}: URGENT! We are trying to contact you. Last weekends draw shows that you have won a £900 prize GUARANTEED. Call 09061701939. Claim code S89. Valid 12hrs only

\textbf{Output}: Spam
\item \textbf{Input}: Aaooooright are you at work?

\textbf{Output}: Ham
\item \textbf{Input}: Please call our customer service representative on FREEPHONE 0808 145 4742 between 9am-11pm as you have WON a guaranteed £1000 cash or £5000 prize!

\textbf{Output}: Spam
\item \textbf{Input}: I'm leaving my house now...

\textbf{Output}: Ham
\item \textbf{Input}: Are you unique enough? Find out from 30th August. www.areyouunique.co.uk

\textbf{Output}: Spam
\item \textbf{Input}: Hello, my love. What are you doing? Did you get to that interview today? Are you you happy? Are you being a good boy? Do you think of me?Are you missing me ?

\textbf{Output}: Ham
\item \textbf{Input}: 500 New Mobiles from 2004, MUST GO! Txt: NOKIA to No: 89545 \& collect yours today!From ONLY £1 www.4-tc.biz 2optout 087187262701.50gbp/mtmsg18

\textbf{Output}: Spam
\item \textbf{Input}: Keep yourself safe for me because I need you and I miss you already and I envy everyone that see's you in real life

\textbf{Output}: Ham
\item \textbf{Input}: Will u meet ur dream partner soon? Is ur career off 2 a flyng start? 2 find out free, txt HORO followed by ur star sign, e. g. HORO ARIES

\textbf{Output}: Spam
\item \textbf{Input}: New car and house for my parents.:)i have only new job in hand:)

\textbf{Output}: Ham
\item \textbf{Input}: Text \& meet someone sexy today. U can find a date or even flirt its up to U. Join 4 just 10p. REPLY with NAME \& AGE eg Sam 25. 18 -msg recd@thirtyeight pence

\textbf{Output}: Spam
\item \textbf{Input}: I'm so in love with you. I'm excited each day i spend with you. You make me so happy.

\textbf{Output}: Ham
\item \textbf{Input}: U 447801259231 have a secret admirer who is looking 2 make contact with U-find out who they R*reveal who thinks UR so special-call on 09058094597

\textbf{Output}: Spam
\item \textbf{Input}: I place all ur points on e cultures module already.

\textbf{Output}: Ham
\item \textbf{Input}: Congratulations ur awarded 500 of CD vouchers or 125gift guaranteed \& Free entry 2 100 wkly draw txt MUSIC to 87066 TnCs www.Ldew.com1win150ppmx3age16

\textbf{Output}: Spam
\item \textbf{Input}: Hi frnd, which is best way to avoid missunderstding wit our beloved one's?

\textbf{Output}: Ham
\item \textbf{Input}: We tried to contact you re your reply to our offer of a Video Handset? 750 anytime networks mins? UNLIMITED TEXT? Camcorder? Reply or call 08000930705 NOW

\textbf{Output}: Spam
\item \textbf{Input}: Great escape. I fancy the bridge but needs her lager. See you tomo 

\textbf{Output}: Ham
\item \textbf{Input}: Hey I am really horny want to chat or see me naked text hot to 69698 text charged at 150pm to unsubscribe text stop 69698

\textbf{Output}: Spam
\item \textbf{Input}: Yes :)it completely in out of form:)clark also utter waste.

\textbf{Output}: Ham
\item \textbf{Input}: Ur ringtone service has changed! 25 Free credits! Go to club4mobiles.com to choose content now! Stop? txt CLUB STOP to 87070. 150p/wk Club4 PO Box1146 MK45 2WT

\textbf{Output}: Spam
\item \textbf{Input}: Sir, I need AXIS BANK account no and bank address.

\textbf{Output}: Ham
\item \textbf{Input}: Ringtone Club: Get the UK singles chart on your mobile each week and choose any top quality ringtone! This message is free of charge.

\textbf{Output}: Spam
\item \textbf{Input}: Hmmm.. Thk sure got time to hop ard... Ya, can go 4 free abt... Muz call u to discuss liao... 

\textbf{Output}: 

Answer: Ham

\end{itemize}

\section{Constraints of Demonstrations}

To mitigate the influence of sensitive settings on ICL performance, including prompting templates, input-label mappings, text distribution, and label space considerations, we ensure that our demonstration examples meet the following constraints(\cite{min-etal-2022-rethinking}):
\begin{itemize}
    \item \textbf{Following unified input and output format:} Ensure that all the $n$ demonstration examples follow a consistent input and output format. 
    \item \textbf{Covering all values in the label space:} Ensure that all the $n$ demonstration examples cover the entire range of possible values in the label space.
    \item \textbf{Selecting from training dataset:} Select the demonstration examples from the training dataset itself. 
    \item \textbf{Forming input-output pair:} Ensure that each of the examples is associated with a corresponding output to create an input-output pair.
\end{itemize}

\section{Additional Experiments}

\subsection{Detailed experiments results}

Table \ref{Table:table2},\ref{Table:table3},\ref{Table:table4},\ref{Table:table5} present the detailed experimental results of the LDP-ICL classification scenario, as illustrated in Figure 3 in the main text.

\begin{table*}[!htbp]
\centering
\begin{tabular}{cccccccccc}
\\ \hline
\multirow{6}{*}{SST-2} & \multicolumn{1}{l}{$\epsilon$=0} & \multicolumn{1}{l}{$\epsilon$=0.5} & \multicolumn{1}{l}{$\epsilon$=1} & \multicolumn{1}{l}{$\epsilon$=2} & \multicolumn{1}{l}{$\epsilon$=3} & \multicolumn{1}{l}{$\epsilon$=8} & \multicolumn{1}{l}{$\epsilon$=$\infty$(ICL)} & \multicolumn{1}{l}{ZSL} & \multicolumn{1}{l}{FL-ICL} \\   \hline
     & 80.00                   & 90.00                     & 91.33                   & 95.33                   & 96.67                   & 94.70                   & 95.33                        & 80.00                         & 36.00                      \\
     & 76.67                   & 86.70                     & 91.33                   & 94.67                   & 94.00                   & 95.33                   & 94.70                        & 74.00                         & 39.33                      \\
     & 76.00                   & 88.70                     & 93.33                   & 95.33                   & 93.33                   & 95.33                   & 96.00                        & 79.33                         & 33.33                      \\
     & 74.67                   & 81.33                     & 90.67                   & 96.00                   & 96.00                   & 94.00                   & 94.67                        & 77.33                         & 42.00                      \\
     & 69.33                   & 76.00                     & 86.67                   & 91.33                   & 91.33                   & 92.00                   & 91.33                        & 72.67                         & 38.00                      \\
     & 76.67                   & 90.67                     & 92.67                   & 92.67                   & 93.33                   & 93.33                   & 92.67                        & 79.33                         & 37.33                      \\ \hline
AVG & 75.56                   & 85.57                     & 91.00                   & 94.22                   & 94.11                   & 94.12                   & 94.12                        & 77.11                         & 37.66                      \\ \hline
STD & 2.98                    & 4.87                      & 1.97                    & 1.54                    & 1.65                    & 1.10                    & 1.49                         & 2.61                          & 2.49                \\ \hline     
\end{tabular}

\caption{Performance of LDP-ICL and baselines across six runs on SST-2($n=32$).}
\label{Table:table2}
\end{table*}

\begin{table*}[!htbp]
\begin{tabular}{cccccccccc}
\\ \hline
\multirow{6}{*}{Subj}   & \multicolumn{1}{l}{$\epsilon$=0} & \multicolumn{1}{l}{$\epsilon$=0.5} & \multicolumn{1}{l}{$\epsilon$=1} & \multicolumn{1}{l}{$\epsilon$=2} & \multicolumn{1}{l}{$\epsilon$=3} & \multicolumn{1}{l}{$\epsilon$=8} & \multicolumn{1}{l}{$\epsilon$=$\infty$(ICL)} & \multicolumn{1}{l}{ZSL} & \multicolumn{1}{l}{FL-ICL} \\   \hline
     & 56.67                   & 66.00                     & 76.67                   & 84.67                   & 79.33                   & 80.67                   & 81.33                        & 58.00                         & 20.67                      \\
     & 59.33                   & 63.33                     & 70.67                   & 82.00                   & 86.67                   & 86.00                   & 84.67                        & 57.30                         & 22.67                      \\
     & 64.00                   & 66.00                     & 72.67                   & 84.00                   & 90.00                   & 90.00                   & 92.67                        & 58.00                         & 12.67                      \\
     & 57.33                   & 66.00                     & 76.67                   & 76.00                   & 84.00                   & 86.67                   & 84.00                        & 56.70                         & 17.33                      \\
     & 60.00                   & 66.67                     & 76.67                   & 76.67                   & 80.00                   & 78.00                   & 82.00                        & 56.70                         & 18.00                      \\
     & 54.00                   & 64.00                     & 71.33                   & 73.33                   & 80.67                   & 86.00                   & 84.00                        & 57.30                         & 21.33                      \\ \hline
AVG & 58.55                   & 65.33                     & 74.11                   & 79.45                   & 83.45                   & 84.56                   & 84.78                        & 57.33                         & 18.78                      \\ \hline
STD & 3.11                    & 1.22                      & 2.62                    & 4.31                    & 3.88                    & 4.01                    & 3.72                         & 0.53                          & 3.30              
\\ \hline
\end{tabular}

\caption{Performance of LDP-ICL and baselines across six runs on Subj($n=32$).}
\label{Table:table3}

\end{table*}

\begin{table*}[!htbp]
\begin{tabular}{cccccccccc}
\\ \hline
\multirow{6}{*}{Ethos}   & \multicolumn{1}{l}{$\epsilon$=0} & \multicolumn{1}{l}{$\epsilon$=0.5} & \multicolumn{1}{l}{$\epsilon$=1} & \multicolumn{1}{l}{$\epsilon$=2} & \multicolumn{1}{l}{$\epsilon$=3} & \multicolumn{1}{l}{$\epsilon$=8} & \multicolumn{1}{l}{$\epsilon$=$\infty$(ICL)} & \multicolumn{1}{l}{ZSL} & \multicolumn{1}{l}{FL-ICL} \\   \hline
     & 88.00                   & 90.67                     & 90.00                   & 94.00                   & 92.00                   & 92.00                   & 92.00                        & 72.67                         & 68.67                      \\
     & 87.07                   & 86.39                     & 80.95                   & 87.07                   & 89.11                   & 88.44                   & 87.07                        & 72.00                         & 68.66                      \\
     & 90.48                   & 90.48                     & 94.56                   & 93.88                   & 93.88                   & 94.56                   & 95.24                        & 82.67                         & 80.00                      \\
     & 94.56                   & 96.60                     & 98.64                   & 97.29                   & 97.96                   & 97.28                   & 98.64                        & 90.67                         & 84.00                      \\
     & 94.56                   & 95.24                     & 95.92                   & 94.56                   & 96.60                   & 95.92                   & 97.96                        & 91.33                         & 79.30                      \\
     & 93.12                   & 94.56                     & 95.24                   & 96.60                   & 96.60                   & 97.96                   & 97.28                        & 88.67                         & 79.30                      \\ \hline
AVG & 91.30                   & 92.32                     & 92.55                   & 93.90                   & 94.36                   & 94.36                   & 94.70                        & 83.00                         & 76.66                      \\ \hline
STD & 3.00                    & 3.49                      & 5.78                    & 3.31                    & 3.06                    & 3.28                    & 4.05                         & 8.04                          & 5.87                  \\ \hline    
\end{tabular}

\caption{Performance of LDP-ICL and baselines across six runs on Ethos($n=32$).}
\label{Table:table4}

\end{table*}

\begin{table*}[!htbp]
\begin{tabular}{cccccccccc}
\\ \hline
\multirow{6}{*}{SMS\_Spam}   & \multicolumn{1}{l}{$\epsilon$=0} & \multicolumn{1}{l}{$\epsilon$=0.5} & \multicolumn{1}{l}{$\epsilon$=1} & \multicolumn{1}{l}{$\epsilon$=2} & \multicolumn{1}{l}{$\epsilon$=3} & \multicolumn{1}{l}{$\epsilon$=8} & \multicolumn{1}{l}{$\epsilon$=$\infty$(ICL)} & \multicolumn{1}{l}{ZSL} & \multicolumn{1}{l}{FL-ICL} \\   \hline
     & 70.67                   & 78.67                     & 84.67                   & 92.00                   & 95.33                   & 94.67                   & 94.67                        & 77.30                         & 7.30                       \\
     & 68.00                   & 79.33                     & 83.33                   & 86.67                   & 90.67                   & 95.33                   & 94.67                        & 78.67                         & 4.70                       \\
     & 71.33                   & 79.33                     & 87.33                   & 91.33                   & 96.00                   & 98.67                   & 96.00                        & 73.33                         & 8.00                       \\
     & 65.33                   & 82.67                     & 88.00                   & 88.67                   & 95.30                   & 96.00                   & 95.33                        & 82.00                         & 10.00                      \\
     & 69.33                   & 74.67                     & 84.67                   & 92.00                   & 92.00                   & 95.33                   & 97.33                        & 80.00                         & 5.30                       \\
    & 76.67                   & 82.67                     & 90.00                   & 90.00                   & 97.33                   & 98.67                   & 97.33                        & 78.67                         & 10.67                      \\ \hline
AVG & 70.22                   & 79.56                     & 86.33                   & 90.11                  & 94.44                   & 96.44                   & 95.89                        & 78.33                         & 7.66                       \\ \hline
STD & 3.48                    & 2.71                      & 2.30                    & 1.94                    & 2.33                    & 1.62                    & 1.11                         & 2.66                          & 2.20                      \\ \hline
\end{tabular}

\caption{Performance of LDP-ICL and baselines across six runs on SMS\_Spam($n=32$).}
\label{Table:table5}

\end{table*}

\subsection{Ablation study}
Our intuition for choosing demonstration examples was to assess whether a model can learn input-label mappings and override semantic priors. A performance below 50\% accuracy in FL-ICL indicates the model's ability to achieve this\cite{wei2023larger}. 

\begin{table*}[!htbp]
\centering
\begin{tabular}{|c|ccccc|}
\hline
\multirow{2}{*}{Dataset} & \multicolumn{5}{c|}{Number of demonstration examples}                                                              \\ \cline{2-6} 
                         & \multicolumn{1}{c|}{4}  & \multicolumn{1}{c|}{8}  & \multicolumn{1}{c|}{16} & \multicolumn{1}{c|}{32} & 64 \\ \hline
SST-2                    & \multicolumn{1}{c|}{55} & \multicolumn{1}{c|}{44} & \multicolumn{1}{c|}{33} & \multicolumn{1}{c|}{31} & 32 \\ \hline
Subj                     & \multicolumn{1}{c|}{93} & \multicolumn{1}{c|}{78} & \multicolumn{1}{c|}{64} & \multicolumn{1}{c|}{38} & 39 \\ \hline
\end{tabular}
\caption{Performance of FL-ICL over number of demonstration examples on SST-2 and Subj.}
\label{Table:FL-ICL}
\end{table*}

Table \ref{Table:FL-ICL} presents the performance for the selection of  $n=32$ demonstration examples, which indicates that the accuracy rate falls below 50\% and hence show  LLM's capability to learn input-label mappings and override semantic priors. Additionally, we carried out ablation studies to analyze how varying the quantity of demonstration examples affects the sentiment analysis performance for the SST-2 task. The analysis was conducted under three demonstration number cases: $n = 16, 32, 64$.

As depicted in Figure \ref{fig:ablation}, we find that these three curves exhibit an identical trend of change regardless of the variation in example quantities. A  variation is observed when the privacy parameter falls within the range of 0 to 0.5: with an increase in the number of examples, accuracy diminishes. Our initial analysis indicates that for a given privacy parameter value, a higher quantity of examples leads to a higher count of flipped examples, which implies a more powerful task-learning ability and hence a less accurate prediction rate. Drawing on the preceding formula (Eq. (5)) in the main text), the heightened presence of flipped examples corresponds to a more pronounced influence on the accuracy.
 \begin{figure}[!htbp]
    \centering
    \includegraphics[width=\linewidth]{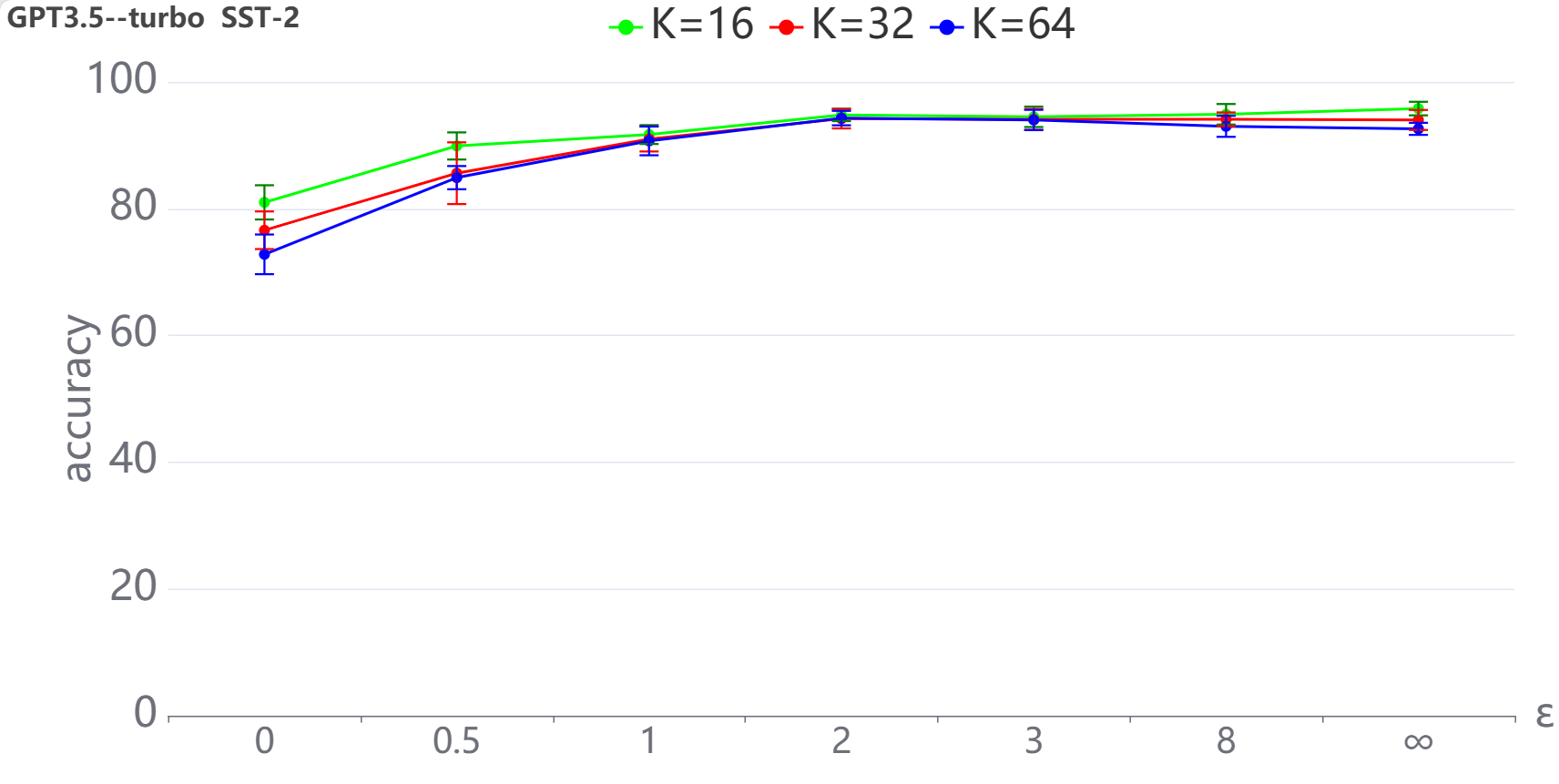}
    \caption{Performance across numbers of the examples}
    \label{fig:ablation}
\end{figure}

Specifically, the detailed results of $n=16$ and $n=64$ are listed in Table \ref{Table:table6} and Table \ref{Table:table7}.

\begin{table*}[!htbp]
\begin{tabular}{cccccccccc}
\\ \hline
\multirow{6}{*}{SST-2}   & \multicolumn{1}{l}{$\epsilon$=0} & \multicolumn{1}{l}{$\epsilon$=0.5} & \multicolumn{1}{l}{$\epsilon$=1} & \multicolumn{1}{l}{$\epsilon$=2} & \multicolumn{1}{l}{$\epsilon$=3} & \multicolumn{1}{l}{$\epsilon$=8} & \multicolumn{1}{l}{$\epsilon$=$\infty$(ICL)} & \multicolumn{1}{l}{ZSL} & \multicolumn{1}{l}{FL-ICL} \\   \hline
   & 83.30                   & 92.00                     & 89.30                   & 96.00                   & 96.70                   & 96.70                   & 97.30                        & 80.00                         & 36.49                      \\
   & 76.70                   & 90.70                     & 92.00                   & 94.00                   & 95.30                   & 95.30                   & 96.70                        & 74.00                         & 38.10                      \\
   & 82.00                   & 90.70                     & 93.30                   & 94.70                   & 94.70                   & 94.70                   & 95.30                        & 79.33                         & 38.78                      \\
  & 78.70                   & 90.00                     & 92.70                   & 96.00                   & 95.30                   & 96.70                   & 96.00                        & 77.33                         & 35.14                      \\
  & 84.70                   & 85.30                     & 92.70                   & 94.70                   & 92.70                   & 92.00                   & 94.00                        & 72.67                         & 32.43                      \\
   & 80.70                   & 90.70                     & 90.00                   & 93.30                   & 92.00                   & 94.00                   & 95.30                        & 79.33                         & 34.00                      \\ \hline
AVG     & 81.02                   & 89.90                     & 91.67                   & 94.78                   & 94.45                   & 94.90                   & 95.77                        & 77.11                         & 35.82                      \\ \hline
STD     & 2.70                    & 2.14                      & 1.49                    & 0.98                    & 1.61                    & 1.63                    & 1.07                         & 2.82                          & 2.22    \\ \hline               
\end{tabular}

\caption{Performance of LDP-ICL and baselines across six runs on on SST-2($n=16$).}
\label{Table:table6}
\end{table*}

\begin{table*}[!htbp]
\begin{tabular}{cccccccccc}
\\ \hline
\multirow{6}{*}{SST-2}   & \multicolumn{1}{l}{$\epsilon$=0} & \multicolumn{1}{l}{$\epsilon$=0.5} & \multicolumn{1}{l}{$\epsilon$=1} & \multicolumn{1}{l}{$\epsilon$=2} & \multicolumn{1}{l}{$\epsilon$=3} & \multicolumn{1}{l}{$\epsilon$=8} & \multicolumn{1}{l}{$\epsilon$=$\infty$(ICL)} & \multicolumn{1}{l}{ZSL} & \multicolumn{1}{l}{FL-ICL} \\   \hline
  & 70.67                   & 81.33                     & 90.67                   & 95.33                   & 94.67                   & 92.00                   & 93.33                        & 58.00                         & 38.67                      \\
  & 75.33                   & 85.33                     & 91.33                   & 94.67                   & 93.33                   & 93.33                   & 91.33                        & 57.30                         & 44.67                      \\
  & 69.33                   & 84.67                     & 94.67                   & 93.33                   & 96.00                   & 94.00                   & 93.33                        & 58.00                         & 44.00                      \\
   & 77.33                   & 86.00                     & 91.33                   & 96.00                   & 96.00                   & 96.00                   & 93.33                        & 56.70                         & 36.00                      \\
  & 69.33                   & 87.33                     & 88.00                   & 94.00                   & 92.66                   & 91.33                   & 93.20                        & 56.70                         & 45.33                      \\
   & 74.67                   & 84.67                     & 88.00                   & 92.67                   & 91.83                   & 91.33                   & 91.22                        & 57.30                         & 44.67                      \\ \hline
AVG     & 72.78                   & 84.89                     & 90.67                   & 94.33                   & 94.08                   & 93.00                   & 92.62                        & 57.33                         & 42.22                      \\ \hline
STD     & 3.14                    & 1.83                      & 2.28                    & 1.14                    & 1.60                    & 1.67                    & 0.96                         & 0.53                          & 3.56    \\ \hline                  
\end{tabular}
\caption{Performance of LDP-ICL and baselines across six runs on on SST-2($n=64$).}
\label{Table:table7}
\end{table*}

\end{document}